%% file: main.tex
\documentclass[lettersize,journal,twoside]{IEEEtran}
\input{utils/preamble}

\input{utils/macros}

\begin{document}

\title{Physically Consistent SINDy (Sparse Identification of Nonlinear Dynamics) for Microgrid Identification and Real-Time Frequency Control}

\author{Mohan Du~\orcidlink{0000-0002-5465-8161},~\IEEEmembership{Student Member,~IEEE}, 
Jiayi Lai, 
Rong-Peng Liu~\orcidlink{0000-0002-7648-3411},~\IEEEmembership{Member,~IEEE}, 
and Xiaozhe Wang~\orcidlink{0000-0002-1887-0990},~\IEEEmembership{Senior Member,~IEEE}
\thanks{Corresponding author: Xiaozhe Wang (e-mail: xiaozhe.wang2@mcgill.ca). 
This work was supported by Natural Sciences and Engineering Research
Council (NSERC) and the Fonds de Recherche du Québec-Nature et technologies under Grants FRQ-NT PR-298827, FRQ-NT 2023-NOVA-314338, FRQ-NT 344058 and Prix Grands Sages Ashok Vijh en électrochimie.
The authors acknowledge the use of ChatGPT (OpenAI), for grammar correction, language refinement, and stylistic suggestions. All original content, ideas, and technical contributions are the sole work of the authors.}}

\newcommand{\IEEEcopyrightyear}{2026}\IEEEaftertitletext{%
    \begin{center}
    \parbox{0.96\textwidth}{%
        \centering
        \footnotesize
        \textcopyright\ \IEEEcopyrightyear\ IEEE.
        Personal use of this material is permitted.
        Permission from IEEE must be obtained for all other uses, in any current or future media, including reprinting/republishing this material for advertising or promotional purposes, creating new collective works, for resale or redistribution to servers or lists, or reuse of any copyrighted component of this work in other works.
        \par
    }
    \end{center}
    \vspace{0.5\baselineskip}
}

\maketitle

\begin{abstract}
This paper proposes PC-SINDYc, a novel framework for the identification and frequency control of microgrids (MGs) with distributed energy resources. By leveraging physics-guided library construction, total least squares regression, and random sample consensus, the regression algorithm of PC-SINDYc robustly identifies the true frequency dynamics of MGs from phasor measurement unit (PMU) data, considering noise, delays, and constraint activations. Based on the identified model, the PC-SINDYc framework further incorporates a model predictive controller (MPC) for real-time frequency control. We prove that, under mild conditions, PC-SINDYc ensures asymptotic stability of the MG.  
Simulations on 4-bus and 13-bus MGs demonstrate that PC-SINDYc effectively controls MG's frequency across various disturbances unseen during the offline identification, outperforming PI controllers, conventional SINDYc, and state-of-the-art reinforcement learning methods.

\end{abstract}

\begin{IEEEkeywords}
Distributed energy resource (DER), 
frequency control, 
grid following, 
grid forming, 
microgrid (MG), 
sparse identification of nonlinear dynamics with control (SINDYc).
\end{IEEEkeywords}

\section{Introduction}\label{sec:introduction}

\IEEEPARstart{M}{icrogrids} (MGs) play a key role in integrating distributed energy resources (DERs), such as solar photovoltaic (PV) and wind power, into main power systems while enhancing the reliability of renewable energy utilization \cite{gong2023a}.
However, MGs face stability challenges due to the stochastic nature of DERs and inertia-less converter interfaces \cite{gong2023d}.
Hierarchical control, comprising primary, secondary, and tertiary controls, is a widely adopted scheme to enhance MG stability \cite{she2023}. 
The primary control stabilizes local voltage and frequency immediately after disturbances; the secondary control eliminates the resulting steady-state deviations; and the tertiary control optimizes long-term power flow for economic operation \cite{she2023}.
This paper focuses on secondary control, which restores MG frequency after large disturbances, as primary control alone cannot fully address frequency deviations \cite{gong2023a}.

A wide range of secondary control methods has been proposed for MGs, broadly classified into model-based and {data-driven} %
approaches. {Model-based} methods, such as those using small-signal analysis \cite{zhang2016} or robust control design \cite{lai2019}, rely on accurate parameters of DERs and network components. In practice, however, such information is often unavailable due to heterogeneous converter types, diverse control modes, and time-varying MG topologies \cite{gong2023a}. To eliminate the reliance on models and parameter values, {data-driven methods have been widely explored. Early approaches include local or distributed heuristic controllers} %
such as proportional–integral (PI) control \cite{savaghebi2012} and consensus-based \cite{yang2024} controls.
However, because these local controllers lack system-wide dynamic awareness, they may exhibit suboptimal performance or even instability when the grid topology or operating conditions change. 

To achieve adaptability to grid changes without relying on detailed system information,  {global data-driven methods} have attracted growing interest \cite{madani2021, ye2010, deshmukh2020, ma2023b, she2023}. These methods identify system dynamics directly from global measurement data and design controllers based on the identified models, thus eliminating the need for prior knowledge and enabling online adaptation via recursive identification. {Linear identification} techniques \cite{madani2021, ye2010} locally linearize MG dynamics, but their accuracy decreases outside the linearization region, e.g., under large disturbances. To capture nonlinear behaviors, {{black-box data-driven methods}} %
such as fuzzy logic \cite{deshmukh2020}, neural networks (NNs) \cite{ma2023b}, and reinforcement learning (RL) \cite{she2023}, have been implemented for secondary control. However, fuzzy control scales poorly \cite{she2023}; NN-based methods are often tailored for prediction rather than control; and RL typically requires extensive, high-quality data and may lack physical interpretability and stability guarantees \cite{gong2023d}. Moreover, their performance is not guaranteed when the online operating conditions or network topology deviate from those encountered during training. 

To address these limitations, {structure-informed data-driven methods \cite{xu2023, yeung2017, ma2024, kandaperumal2022, gong2023, gong2023a} have gained attention.} %
Many of these are based on Koopman operator theory \cite{Brunton_Kutz_2019k}. This theory demonstrates that a nonlinear dynamical system can be represented via an infinite-dimensional linear operator acting on a space of candidate functions (observables) \cite{Brunton_Kutz_2019k}. 
However, practical implementations require finite-dimensional approximations. Along these lines, deep-learning-based Koopman methods \cite{xu2023, yeung2017} learn observables implicitly via NNs. However, because these observables lack explicit physical meaning, such methods generally lack performance guarantees and are therefore remain closer to {black-box data-driven methods}.
In contrast, analytical library constructed from domain knowledge \cite{ma2024} improves interpretability but %
requires extensive knowledge of grid and converter parameters, which limits its practicality. 
\cite{gong2023, gong2023a, kandaperumal2022} include polynomial, trigonometric, and lagged time-series functions of voltage angles and frequencies as candidate functions in the library and applied it for identification and control in DER-based MGs. Leveraging sparse identification of nonlinear dynamics (SINDy) \cite{brunton2016a}, \cite{nandakumar2023} constructs a polynomial-trigonometric library to predict MG voltages and frequencies. 
Although these physics-inspired, yet intuitively constructed libraries improve  interpretability and may better capture MG dynamics, two main challenges remain. First, they do not necessarily identify the \textit{true} underlying physical model, limiting global accuracy and generalizability. This can lead to degraded prediction and control performance under disturbances never seen during identification (e.g., \cite{kandaperumal2022,nandakumar2023}) or requires substantial real-time computation to track fast dynamics (e.g., {\cite{gong2023, gong2023a}}). Second, the original SINDy is sensitive to noise, and its identification accuracy can deteriorate when converter control limiters become active.

To address these challenges, we propose a framework, named \emph{physically consistent} SINDy with control (PC-SINDYc) to identify the \textit{true} MG dynamics model for frequency dynamics and enable real-time frequency control. Our key advantages are:\\

\noindent$\bullet$ \textit{Physics-guided library construction}: We propose a systematic and physics-guided method for constructing the candidate function library that accounts for both grid-forming (GFM) and grid-following (GFL) converters under different control modes. This ensures that the identified SINDy model remains equivalent to the \emph{true} governing equations even under large-signal transients. %
\\
\noindent$\bullet$ \textit{Minimally intrusive and purely data-driven:} PC-SINDYc requires only small-signal excitation of the system ($\pm0.001$ to $\pm0.01$~p.u.) and 10~s of phasor measurement unit (PMU) data, even in the presence of noise, to identify the MG frequency dynamics model. It does not require parameters of network or converters.\\
\noindent$\bullet$ \textit{Robust identification under noise and constraints:} PC-SINDYc enhances the original SINDy by integrating total least squares (TLS) regression~\cite{vanHuffel1991} to mitigate regression dilution under measurement noise, and random sample consensus (RANSAC)~\cite{fischler1981} to automatically reject corrupted data affected by constraint activations in converters. Moreover, we provide a \textit{probabilistic convergence guarantee} to demonstrate its computational efficiency. \\
\noindent$\bullet$ \textit{Guaranteed control performance:} Based on the identified model, we design an online model predictive controller (MPC) that can perform frequency control under large disturbances never seen during the offline identification/training process. We prove that, under mild conditions, PC-SINDYc guarantees \textit{asymptotic stability} of the MG. 

This paper is organized as follows.
Section~\ref{sec:DER-dynamics} introduces frequency dynamics in DER-based MGs. 
Section~\ref{sec:sindy} reviews the conventional SINDy method. 
Section~\ref{sec:construct-library} constructs an analytically derived candidate function library. 
Section~\ref{sec:methodology} details the data acquisition methodology, the PC-SINDYc framework, and the proof of asymptotic stability. 
Section~\ref{sec:simulation} presents case studies. 
Section~\ref{sec:conclusion} concludes the paper.

\textbf{Notations.} 
We denote by $\bbn$ the set of natural numbers and $\bbn_{0:n}\coloneqq \{0,1,\ldots,n\}$.
Boldface letters, e.g., $\bx\in \mathbb{R}^{n_x}$, denote column vectors, where $n_{(\cdot)}$ denotes the dimension of input vector. 
For a time-dependent vector $\bx_t \in \mathbb{R}^{n_x}$, we omit the time index and write $\bx$ when it is clear from context. 
We denote by $\bx^{+}$ the value of $\bx$ at the next sampling instant in a discrete-time system.
For any two vectors $\bx,\by$, $[\bx;\by] \ceq [\bx^\top, \by^\top]^\top$ denotes their column concatenation.
$\bx[i]$ denotes the $i$-th component of vector $\bx$, and $\bx[\mathcal{I}] \ceq (\bx[i])_{i \in \mathcal{I}}$ for index set $\mathcal{I}$. 
We denote by $\bX[i,j]$ the $(i,j)$-th entry of matrix $\bX$, and by $\bX[i,:]$ and $\bX[:,j]$ its $i$-th row and $j$-th column, respectively; furthermore, $\bX[\mathcal{I},:] \ceq (\bX[i,:])_{i \in \mathcal{I}}$ and $\bX[:,\mathcal{J}] \ceq ((\bX[:,j]^\top)_{j \in \mathcal{J}})^\top$. 
$\|\cdot\|$ denotes the 2-norm of a vector or the induced 2-norm for matrices, and $\|\bx\|_Q \ceq \sqrt{\bx^\top Q\bx}$ with $Q\succeq 0$. 

\section{Preliminary: Frequency Dynamics of DER-Based MGs}\label{sec:DER-dynamics}
This section examines the general modeling and control modes of the two widely used types of DERs: GFM and GFL converters. We consider an MG comprising both GFM and GFL converters, distributed across the sets of buses \(\cn_\mgfm\) and \(\cn_\mgfl\), respectively, with the overall bus set defined as \(\cn \ceq \cn_\mgfm \cup \cn_\mgfl\).

\subsection{GFM Converter}
A GFM converter is responsible for establishing the MG frequency \cite{she2023}. 
{Two widely used GFM control modes, droop control and virtual synchronous generator (VSG) control~\cite{she2023}, are shown in Figs.~\ref{fig:GFM_blocks}(b) and (c) for units $i \in \cn_{\mgfm}^\mdroop$ and $i \in \cn_{\mgfm}^\mvsg$, respectively, where $\cn_\mgfm = \cn_\mgfm^\mdroop \cup \cn_\mgfm^\mvsg$.}

\begin{figure}[htbp]
    \centering
    \includegraphics[width=0.92\linewidth]{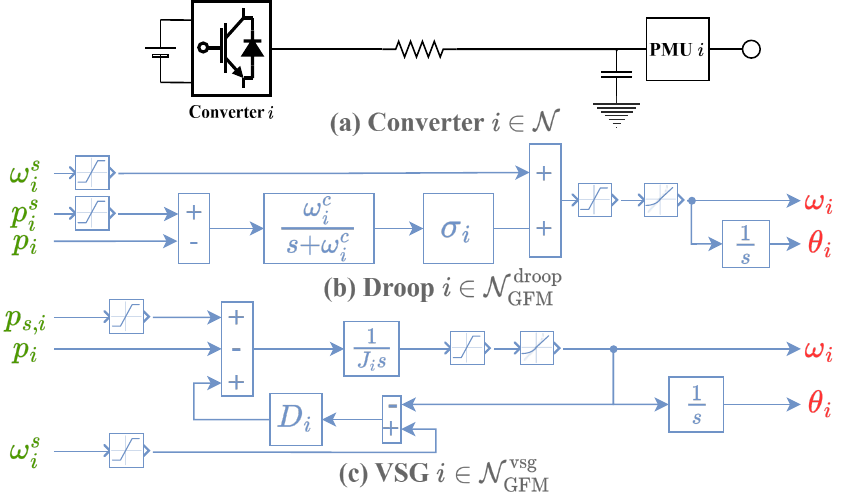}
    \caption{(a) Converter and PMU $i$; (b) Droop control and (c) VSG control of GFM converters \cite{she2023}}
    \label{fig:GFM_blocks}
\end{figure}

\subsubsection{Droop Control:}
The droop control adjusts the output frequency $\omega_i$ proportionally to the active power deviation $p_i-p^{s}_{i}$ from its setpoint $\omega^{s}_{i}$, where $p_i$ and $p^{s}_{i}$ are the output active power and its setpoint, respectively. 
A first-order low-pass filter is often implemented to attenuate high-frequency noise in the measured output power \cite{Yazdani2010}, as shown in Fig. \ref{fig:GFM_blocks}(b), leading to the dynamics:
\begin{subequations}\label{equ:gfm-droop}
    \begin{flalign}
         \dot{\theta}_{i}&=\omega_{i},  && {i\in\cn_\mgfm^\mdroop,} \label{equ:gfm-droop-1}\\
         \dot{\omega}_{i}&= \omega^{c}_{i} (\omega^{s}_{i}-\omega_{i}) + \omega^{c}_{i} \sigma_{i} (p^{s}_{i}-p_{i}), && {i\in\cn_\mgfm^\mdroop}, \label{equ:gfm-droop-2}
    \end{flalign}
\end{subequations}
where \(\sigma_{i}\) is the slope of the droop curve and $\omega^c_i$ is the cut-off frequency of the first-order low-pass filter. 

\subsubsection{VSG Control:}
The dynamics of a VSG control, as shown in Fig. \ref{fig:GFM_blocks}(c), can be represented as:
\begin{subequations}\label{equ:gfm-vsg} 
    \begin{flalign}
        \dot{\theta}_{i}& =\omega_{i}, && {i\in\cn_\mgfm^\mvsg}, \label{equ:gfm-vsg-1}\\
        \dot{\omega}_{i}& = {D_i}/{J_{i}}\cdot (\omega^{s}_{i}-\omega_{i}) + {1}/{J_{i}}\cdot(p^{s}_{i}- p_{i}), \!\!\!&& {i\in\cn_\mgfm^\mvsg}, \label{equ:gfm-vsg-2}
    \end{flalign}
\end{subequations}
where the damping coefficient \(D_{i}\) and inertia constant \(J_{i}\) mimic the damping factor and rotational inertia of synchronous generators, respectively. If we let
    \begin{align}
        & J_i={1}/({\omega^{c}_{i} \sigma_{i}}), \,  D_i={1}/{\sigma_i}, \qquad &{i\in\cn_\mgfm,} 
    \end{align}
the VSG dynamics in~\eqref{equ:gfm-vsg} are mathematically equivalent to the droop dynamics \cite{darco2014}. 
{As such, we define the state vector \(\bx_i\) and control input \(\bu_i\) of GFM \(i\) as
\begin{align}\label{equ:gfm-droop-xu}
    \bx_i \ceq  [\theta_i, \omega_i]^\top, \, \bu_i \ceq  [\omega^{s}_{i}, p^{s}_{i}]^\top, \qquad &i\in\cn_\mgfm. %
\end{align}}

\subsection{GFL Converter}\label{sec:gfl}
A GFL converter relies on the phase-locked loop (PLL) for frequency synchronization \cite{Yazdani2010}.
For GFL $i\in\cn_{GFL}$, two widely used PLL implementations---the three-phase and single-phase PLLs---are shown in Figs. \ref{fig:GFL_blocks}(a) and (b), respectively. 
As shown in Fig.~\ref{fig:GFL_blocks}, the main difference between the three-phase and single-phase PLLs lies in the calculation of $\delta_i$, which indicates the deviation between the PLL-estimated phase $\theta_{i}$ and the real phase of the terminal voltage $v^{abc}_i$. 

\begin{figure}[htbp]
    \centering
    \includegraphics[width=0.92\linewidth]{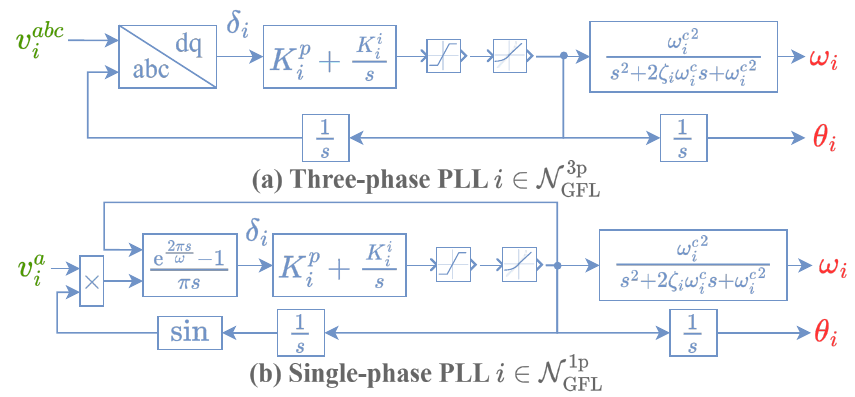}
    \caption{(a) Three-phase PLL~\cite{matlab-pll3p} and (b) single-phase PLL~\cite{golestan2019}}
    \label{fig:GFL_blocks}
\end{figure}
\subsubsection{Three-Phase PLL:}
The three-phase PLL obtains $\delta_i$ by performing a q-axis Park transformation 
$\mathbf{T}_q({\theta}_{i})$ to the three-phase terminal voltage \(\bm{v}^{abc}_{i}\ceq [v^{a}_{i};v^{b}_{i};v^{c}_{i}]\): 
\begin{subequations}\label{equ:3p-pll-vq}
\begin{flalign}
& \mathbf{T}_q({\theta}_{i})\ceq -{2}/{3}\left[
\sin({\theta}_{i}),\sin\left({\theta}_{i} - {2\pi}/{3}\right) ,  \sin\left({\theta}_{i} + {2\pi}/{3}\right)\right], \hspace*{-10em}&& \nonumber\\
& &&{i\in\cn_\mgfl^\mtp,}\label{equ:3p-pll-vq-T}\\
& \delta_{i} = \mathbf{T}_{q}({\theta}_i)\, \bm{v}^{abc}_{i}, \hspace*{-10em}&&{i\in\cn_\mgfl^\mtp,}\label{equ:3p-pll-vq-vq}
\end{flalign}
\end{subequations} 
\subsubsection{Single-Phase PLL:}
The single-phase PLL obtains $\delta_i$ by averaging the product of the phase-A voltage $v^{a}_{i}$ and the estimated quadrature component \(\sin(\theta_{i})\) over a period $T(\theta_i)$:
\begin{subequations}\label{equ:1p-pll-vq}
\begin{flalign}
    &T(\theta_i(t))\ceq {2\pi}/({\dot{\theta}_i(t)}), &&{i\in\cn_\mgfl^\msp,}\label{equ:1p-pll-vq-T}\\
    &\delta_{i} (t) \!= \!-\frac{2}{T(t)}\int_{t-T(t)}^{t} \!\!v^{a}_{i}(\tau) \sin(\theta_{i}(\tau)) \, \mathrm{d}\tau,  \hspace*{-10em}&&{i\in\cn_\mgfl^\msp.}\label{equ:1p-pll-vq-vq}
\end{flalign}
\end{subequations}
{Under balanced three-phase voltage conditions and when the PLL operates in the lock-in region, \eqref{equ:3p-pll-vq-vq} is equivalent to \eqref{equ:1p-pll-vq-vq}. The proof is given in Appendix~\ref{app:PLL-equiv}.}

After obtaining $\delta_i$, both the three-phase and single-phase PLLs feed \(\delta_{i}\) into a PI controller and a second-order low-pass filter to track the terminal frequency \({\omega}_{i}\) \cite{Yazdani2010, matlab-pll3p}:
\begin{subequations}\label{equ:3p-pll} 
    \begin{flalign}
        &\dot{\theta}_{i} \!=\! K^{p}_{i} \delta_{i} + K^{i}_{i} \textstyle\int \delta_{i}, &&{i\in\cn_\mgfl,\!\!} \label{equ:3p-pll-1}\\ 
        &\ddot{\omega}_{i} \!=\! -2\zeta_{i}\omega^{c}_{i} \dot{\omega}_{i} \!-\! {(\omega^{c}_{i})}^2 (\omega_{i} \!-\! K^{p}_{i} \delta_{i} \!-\! K^{i}_{i} \!\textstyle\int\! \delta_{i}), \hspace*{-10em}&&{i\in\cn_\mgfl,\!\!}\label{equ:3p-pll-2}
    \end{flalign}
\end{subequations}
where $K^{p}_{i}$ and $K^{i}_{i}$ denote the proportional and integral gains of the PI controller, respectively; and $\omega^{c}_{i}$ and $\zeta_{i}$ are the natural frequency and damping ratio of the second-order low-pass filter, respectively. 
Here, \(\int \delta_i\) is used as a shorthand for \(\int_{0}^{t} \delta_i(\tau) \,\mathrm{d}\tau\), for brevity.
{We define the state vector $\bx_{i}$ and control input $\bu_i$ of GFL $i$ as:
\begin{align}\label{equ:pll-xu}
    &\bx_i \ceq  [\theta_i,\; \omega_i,\; \dot{\omega}_i]^\top, \,
    \bu_i \in\emptyset, \qquad i\in\cn_\mgfl. 
\end{align} }

\subsection{Constraints on State and Control Variables}\label{sec:constraints}
According to IEEE Standard C37.118.1-2011~\cite{IEEE2011}, the terminal frequency {$\omega_i$} and control variables {$\omega^s_i$ and $p^s_i$} of converters are constrained within predefined safety limits, which are typically enforced by saturation blocks and rate limiters, as illustrated in Figs.~\ref{fig:GFM_blocks}-\ref{fig:GFL_blocks}. 
Accordingly, the state and control variables satisfy $\bx_i \in \mathcal{X}_i$ and $\bu_i \in \mathcal{U}_i$, where:
\begin{subequations}\label{equ:constraint}
    \begin{flalign}
    \cx_i &\ceq  \{ \bx_i \mid \underline{\omega}_i < \omega_i < \bar{\omega}_i, \underline{\dot{\omega}}_i < \dot{\omega}_i < \bar{\dot{\omega}}_i \}, 
    \hspace*{-1.2em}&& i\in\cn, \hspace*{-0.4em} \label{equ:constraint-x}\\
    \cu_i &\ceq 
        \{ \bu_i \mid \underline{\omega}^{s}_{i} < \omega^{s}_{i} < \bar{\omega}^{s}_{i},
    \underline{p}^{s}_{i} < p^{s}_{i} < \bar{p}^{s}_{i} \}, 
    \hspace*{-1.2em}&& i\in \cn_{\mgfm}, \hspace*{-0.4em}\label{equ:constraint-ugfm}\\
    \cu_i &\ceq \emptyset, 
    &&i\in \cn_{\mgfl}. \hspace*{-0.4em}\label{equ:constraint-ugfl}
    \end{flalign}
\end{subequations}
Here, $\underline{(\cdot)}$ and $\bar{(\cdot)}$ denote the lower and upper bounds of the corresponding variables, respectively.

\begin{remark}\label{remark:constraint}
The differential equations \eqref{equ:gfm-droop}, \eqref{equ:gfm-vsg}, and \eqref{equ:3p-pll} of GFM and GFL converters are valid only when state and control variables stay strictly within the interior of the constraint sets, i.e., $\bx_i \in \cx^{\circ}_i$ and $\bu_i \in \cu^{\circ}_i$ for all $i \in \cn$. 
When a state or control variable reaches the boundary of a constraint set, i.e., $\bx_i \in \partial \cx_i$ or $\bu_i \in \partial \cu_i$, it activates the saturation block or rate limiter. In such a case, the converter $i$ remains operational but no longer follows the dynamics described by \eqref{equ:gfm-droop}, \eqref{equ:gfm-vsg}, or \eqref{equ:3p-pll}. The state and control variables reaching the boundaries of constraint sets are referred to as \textit{outliers} and will be identified and excluded as explained in Section \ref{sec:method-regression}. 
\end{remark}

\section{Conventional SINDy}\label{sec:sindy}
Consider a nonlinear dynamical system 
\begin{equation}
     \dot{\bx}=\bm{f}(\bx, \bu), \label{equ:ODEs}
\end{equation}
where $\bx\in \cx\subseteq \mathbb{R}^{n_{\bx}}$ denotes the system state $\bx$ constrained to the operational set $\cx$ and $\bu\in \cu\subseteq \mathbb{R}^{n_{\bu}}$ denotes the input vector $\bu$ constrained to the admissible set $\cu$. 
As observed in \cite{brunton2016a}, most physical systems have only a few dominant terms in their governing equations, rendering these equations sparse in a high-dimensional space of nonlinear functions. 
In light of this, the nonlinear system \eqref{equ:ODEs} can be reformulated as 
\begin{equation}\label{equ:sindy-vector}
    \dot{\bx} = \Xi^\top \btheta(\bx, \bu),
\end{equation}
where \(\btheta(\bx,\bu):\cx\times\cu\to\mathbb{R}^{n_{\btheta}}\) is the library of candidate functions, and \(\Xi\in\mathbb{R}^{n_{\btheta}\times n_{\bx}}\) is a column-sparse coefficient matrix. 
The objective of SINDy is to design an appropriate library \( \btheta \) such that we could transform the original nonlinear system \eqref{equ:ODEs} into a linear one \eqref{equ:sindy-vector}, thereby greatly simplifying both system identification and controller design. 

In power systems, a common approach to construct a library is to {intuitively} include constant, polynomial, and trigonometric functions into it \cite{gong2023a, cai2023, brunton2016a, nandakumar2023,cai2025}: 
\begin{equation}\label{equ:library-intuitive}
\btheta^{{\mint}}(\bx,\bu)\ceq \begin{bmatrix}
1; \bm{P}_k(\bx,\bu); \sin(\bx,\bu); \cos(\bx,\bu); \cdots\\
\end{bmatrix}, 
\end{equation}
where $\bm{P}_k(\bx,\bu)$ denotes the vector of all monomials in $\bx$ and $\bu$ with total degree up to $k$~\cite{yi2022}.
{However, $\btheta^\mint$ may not capture the true  GFM or GFL frequency dynamics in \eqref{equ:gfm-droop}, \eqref{equ:gfm-vsg}, and \eqref{equ:3p-pll}.
}

{Alternatively, ML methods can generate libraries with implicitly represented functions, greatly extending applicability to a broader class of systems \cite{xu2023, yeung2017}. 
However, the implicit representations lack physical interpretability, limiting controller design with stability guarantees.}

{In summary, both approaches struggle to identify the \emph{true} dynamical model, and may therefore lose generalization ability to unseen disturbances.}
In the next section, we will employ the domain knowledge of DER dynamics to construct a physically consistent library $\btheta$ that can identify the \emph{true} frequency dynamics model of a DER-based MG from a finite training dataset.

\section{Physically Consistent Library of Candidate Functions} \label{sec:construct-library}

We reconstruct the frequency dynamics of GFM \eqref{equ:gfm-droop}--\eqref{equ:gfm-vsg} and GFL \eqref{equ:3p-pll} into the SINDy form
\begin{equation}\label{equ:sindy-i}
    \dot{\bx}_i = \Xi_i^\top \btheta_i(\bx_i, \bu_i),
\end{equation}
and develop the library $\btheta_i$ and coefficient matrix $\Xi_i$ accordingly.
Taking the GFM droop dynamics \eqref{equ:gfm-droop} as an example, we reconstruct it as
\begin{subequations}
\begin{flalign}
    \eqref{equ:gfm-droop} & \Leftrightarrow 
    \begin{bmatrix}\dot{\theta}_i\\ \dot{\omega}_i \end{bmatrix}
    =
    \begin{bmatrix}\omega_{i}\\ \omega^{c}_{i} (\omega^{s}_{i}-\omega_{i}) + \omega^{c}_{i} \sigma_{i} (p^{s}_{i}-p_{i})  \end{bmatrix}, \hspace*{-3cm}&&\hspace*{0.78cm}i\in\cn_\mgfm^\mdroop, \hspace*{-0.78cm} \notag \\
    & \Leftrightarrow 
    \underbrace{\begin{bmatrix}\dot{\theta}_i\\ \dot{\omega}_i \end{bmatrix}}_{\bdx_i}
    =
    {\underbrace{\begin{bmatrix}
    1 & -\omega^{c}_{i} \\
    0 & \omega^{c}_{i} \\
    0 & -\omega^{c}_{i} \sigma_{i} \\
    0 & \omega^{c}_{i} \sigma_{i}
\end{bmatrix}}_{\ceq \Xi_i} }^\top
    \underbrace{\begin{bmatrix}\omega_{i}\\ \omega^{s}_{i} \\ p_{i}\\ p^{s}_{i}\end{bmatrix}}_{\ceq \btheta_i}, &&i\in\cn_\mgfm^\mdroop , \label{equ:gfm-droop-define}\\
    & \Leftrightarrow \bdx_i=\Xi^\top_i \btheta_i &&i\in\cn_\mgfm^\mdroop , \label{equ:gfm-droop-sindy}
\end{flalign}
\end{subequations}
where \eqref{equ:gfm-droop-define} naturally defines $\btheta_i$ and $\Xi_i$ 
such that \eqref{equ:gfm-droop} takes the SINDy form \eqref{equ:sindy-i}.

Similarly, we develop $\btheta_i$ and $\Xi_i$ for GFM VSG~\eqref{equ:gfm-vsg} and GFL PLL~\eqref{equ:3p-pll}, and summarize all cases as:

\begin{subequations}\label{equ:theta-xi-all}
    \begin{flalign}
    &\btheta_i(\bx_{i}, \bu_{i})\ceq  
    \begin{cases}
        \begin{bmatrix}\omega_{i}& \omega^{s}_{i} & p_{i}& p^{s}_{i}\end{bmatrix}^{\!\top}, & i\in\cn_\mgfm,\\
        \begin{bmatrix} \omega_{i}& \dot{\omega}_{i}& \delta_{i}& \textstyle\int \delta_{i}\end{bmatrix}^\top, & i\in\cn_\mgfl,
    \end{cases} \hspace*{-10em}&& \label{equ:btheta_i}\\
    &\Xi_i\!\ceq \!
    \begin{cases}
        \!\!\begin{bmatrix}
        1 & 0 & 0 & 0 \\
        -\omega^{c}_{i} & \omega^{c}_{i} & -\omega^{c}_{i} \sigma_{i} & \omega^{c}_{i} \sigma_{i}
        \end{bmatrix}^{\!\top}\!\!\!\!,  &i\in\cn_\mgfm^\mdroop,
        \\
        \!\!\begin{bmatrix} 
        1 & 0 & 0 & 0 \\ 
        -\frac{D_{i}}{J_{i}} & \frac{D_{i}}{J_{i}} & -\frac{1}{J_{i}} & \frac{1}{J_{i}} 
        \end{bmatrix}^{\!\top}\!\!\!\!, &i\in\cn_\mgfm^\mvsg,
        \\
        \!\!\begin{bmatrix}
            0 & 0 & K^{p}_{i} & K^{i}_{i} \\
            0 & 1 & 0 & 0 \\
            -{(\omega^{c}_{i})}^2 \!\!&\!\! -2\zeta_{i}\omega^{c}_{i} \!\!&\!\! {(\omega^{c}_{i})}^2 K^{p}_{i} \!\!&\!\! {(\omega^{c}_{i})}^2 K^{i}_{i}
        \end{bmatrix}^{\!\!\!\top}\!\!\!, \!\!\! &i\in\cn_\mgfl. 
    \end{cases} &&  \label{equ:xi_i}
    \end{flalign}
\end{subequations}
Stacking the variables for the entire MG yields:
\begin{subequations}\label{equ:whole-defs}
\begin{align}
   &\btheta(\bx,\! \bu) \!\!\cceq\! (\btheta_i(\bx_i,\! \bu_i))_{i\in\cn}, \hspace*{-0.3em}
   & \Xi &\!\cceq\! \diag((\Xi_i)_{i\in\cn}), \label{equ:stack-btheta}\\
   &\bx \!\cceq\! (\bx_i)_{i\in\cn}, 
   & \bu &\!\cceq\! (\bu_i)_{i\in\cn}, \label{equ:bx-bu}\\
   &\cx \!\cceq\!\textstyle \prod_{i\in\cn}\cx_i, 
   & \cu &\!\cceq\! \textstyle\prod_{i\in\cn}\cu_i,
\end{align}
\end{subequations}
where $\prod_{i\in\cn}$ denotes the Cartesian product of sets.

It should  be noted that the parameters in $\Xi$ (e.g., the cut-off frequency $\omega^c_i$ for GFM droop control $i\in\cn^\mdroop_\mgfm$ and GFL PLL $i\in\cn_\mgfl$ and damping coefficient $D_i$ for GFM VSG control $i\in\cn^\mvsg_\mgfm$) 
are typically unknown to utility operators, while the measurements needed to construct library $\btheta$ (e.g., phase angle $\theta_i$ and frequency $\omega_i$ for DER $i\in\cn$) are available for operators via PMUs. Our goal is to use the available PMU measurements to identify the unknown coefficient matrix $\Xi$ that describes the frequency dynamics of the DER-based MG.

This domain-informed formulation allows us to build the libraries \eqref{equ:btheta_i} directly from first principles %
rather than relying on heuristic choices such as 
the generic polynomial–-trigonometric library {$\btheta^\mint$ \eqref{equ:library-intuitive} or ML-based implicit representations. As a result, when combined with advanced regression techniques, the identified coefficient matrices $\hat{\Xi}_i$ and the proposed libraries $\btheta_i$ \eqref{equ:btheta_i} recover the \emph{true} underlying physical model, generalizing to large-signal transients unseen during training and enabling accurate prediction and control beyond the capabilities of $\btheta^\mint$ or ML-based methods.}

\begin{remark}\label{remark:constraint-identification}
 As noted in Remark \ref{remark:constraint}, \eqref{equ:gfm-droop}, \eqref{equ:gfm-vsg}, and \eqref{equ:3p-pll} accurately describe the dynamics of GFM and GFL converters when \eqref{equ:constraint} holds, i.e., when constraint blocks in Figs. \ref{fig:GFM_blocks}(b)-(c) and Fig. \ref{fig:GFL_blocks} are inactivate. Under this condition, the SINDy-form model \eqref{equ:sindy-i} {with $\btheta_i$ and $\Xi_i$ defined }in \eqref{equ:theta-xi-all} is equivalent to \eqref{equ:gfm-droop}, \eqref{equ:gfm-vsg}, and \eqref{equ:3p-pll}, representing the true physical model. 

\end{remark}

\begin{remark}\label{remark:other-blocks}

The dynamics of control blocks other than frequency blocks in Figs. \ref{fig:GFM_blocks}(b)–(c) and \ref{fig:GFL_blocks} (e.g., current and voltage control blocks) are naturally captured by the library $\btheta_i$. Specifically, for GFM $i\in\cn_\mgfm$, these dynamics are mainly reflected in $p_i$; for GFL $i\in\cn_\mgfl$, they are mainly reflected in $v^{abc}_i$ and, therefore, in $\delta_i$. 
Likewise, changes in network topology, triggers of protections, or the activation of constraints in non-frequency blocks are all reflected in the constructed library $\btheta$. 
In sum, when \eqref{equ:constraint} is satisfied, the SINDy-form model \eqref{equ:sindy-i} describes the true physical model despite external limiter activation and changes in operating conditions or network topology.  When \eqref{equ:constraint} is violated, the measurements are termed as \textit{outliers} and will be identified and excluded as explained in Section \ref{sec:method-regression}.    

\end{remark}

In the next section, we explain how the library \(\btheta\) and $\dot{\bm{x}}$ are constructed from PMU measurements and how $\Xi$ is accurately estimated despite measurement noise and constraint activation.  

\section{PC-SINDYc Framework for System Identification and Frequency Control} \label{sec:methodology}

PC-SINDYc first estimates the coefficient matrix $\Xi$ offline using PMU measurements collected under small perturbations. We establish a probabilistic convergence guarantee for this identification process in the presence of measurement noise and constraint activation. Using the offline-identified model $\dot{\bx} = \Xi^\top \btheta(\bx, \bu)$, we then design an online MPC for real-time frequency control and show that, under mild conditions, the resulting closed-loop MG is asymptotically stable.

\subsection{Small Perturbations}\label{sec:method-activation}
To excite the MG frequency dynamics, we inject small and continuous disturbances through the control inputs $\bu$ \eqref{equ:bx-bu}. 

Specifically, we design the disturbance function $\psi(t)$ as a sum of sinusoids with randomized amplitudes, frequencies, and phases:
\begin{flalign}\label{equ:psi}
    &\psi(t; A^\psi, N^\psi, \bar{\omega}^{\psi})
    \ceq \frac{A^\psi}{N^\psi}\sum_{j=1}^{N^\psi} A^\psi_{j}
    \sin(\omega^\psi_{j} t + \delta^\psi_{j}), \hspace*{-20cm}&&\notag\\
    & &&\hspace*{-2cm}A^\psi_{j}\! \sim \!\operatorname{Unif}(-1,1), 
    \omega^\psi_{j}\! \sim \!\operatorname{Unif}(0,\bar{\omega}^{\psi}) ,
    \delta^\psi_{j}\! \sim \!\operatorname{Unif}(0,2\pi). 
\end{flalign}
Here, $A^\psi$ is the desired magnitude of disturbance, $N^\psi$ is the number of sinusoids to be summed, $\bar{\omega}^{\psi}$ is the disturbance bandwidth, and $\operatorname{Unif}(\cdot,\cdot)$ denotes the uniform distribution 

In our implementation, during the identification period $t \in \ct_\mtrn \ceq [0, T_\mtrn]$, we perturb the control input $\bu=\left([\omega^s_i;p^s_i]\right)_{i\in\cn_{\mgfm}}$ with small-signal disturbances
\begin{subequations}\label{equ:excitation}
    \begin{flalign}
        \omega^s_{i,t} &= 1 + \psi(t; 0.001, 50, 50\pi), \hspace*{-4cm} &&i\in\cn_{\mgfm}, t\in\ct_\mtrn,\label{equ:excitation-omega}\\
        p^s_{i,t} &= 0.5 + \psi(t; 0.01, 50, 50\pi), \hspace*{-4cm} &&i\in\cn_{\mgfm}, t\in\ct_\mtrn,\label{equ:excitation-p}
    \end{flalign}
\end{subequations}
in per unit value, so that the excitations remain small and introduce only negligible perturbations around the nominal operating point, and do not activate the constraints \eqref{equ:constraint} for either GFM or GFL converters.

\subsection{Using PMU Measurements to Construct the Library of Candidate Functions ${\vartheta}$}\label{sec:PMU}

With advances in PMU technology, PMU data have been widely used to estimate static and dynamic power-system models. For instance,  using PMU measurements, prior work has identified equivalent static power-flow models \cite{yu2018}, pseudo-dynamic network models \cite{li2018}, dynamic power-flow models \cite{wang2020,pierrou2021,guo2021}, %
and dynamic load models \cite{ge2015,du2021} on sub-second to second timescales.

To construct $\btheta$ and state derivatives $\bdx$, we exploit widely deployed PMU measurements. 
According to IEEE Standard C37.118.1-2011~\cite{IEEE2011}, a PMU at bus $i\in\cn$ is a device that measures three-phase voltage $\bm{v}^{abc}_{i,t}$ and current $\bm{i}^{abc}_{i,t}$ over the sampling time set $t\in\ct^{\ms} \ceq  \left\{ t_k = k \Delta T^{\ms} \mid k \in \mathbb{N} \right\}$ and reports time-synchronized estimates of voltage and current phasors, $V_{i,t}\angle\theta_{i,t}$ and $I_{i,t}\angle\varphi_{i,t}$, terminal frequency $\omega_{i,t}$, and the rate of change of frequency (RoCoF) $\dot{\omega}_{i,t}$ over the reporting time set $t\in\ct^{\mr} \ceq  \left\{ t_k = k \Delta T^{\mr} \mid k \in \mathbb{N} \right\}$. 
{Typically, the reporting time $\Delta T^\mr$ is an integer multiple of the sampling time $\Delta T^\ms$, i.e., $\ct^\mr \subset \ct^\ms$. }

We aim to construct the data stream $\{\btheta(\bx_t, \bu_t), \bdx_t\}_{t \in \ct^\mr}$ over the reporting time set $\ct^\mr$. 
From standard PMU outputs and known control inputs $\bu$, we directly obtain
$\{\omega_{i,t}, \omega^s_{i,t}, p^s_{i,t}, \dot{\omega}_{i,t} \}_{i \in \cn_{\mgfm}, t \in \ct^\mr}, \,
\{\omega_{i,t}, \dot{\omega}_{i,t} \}_{i \in \cn_{\mgfl}, t \in \ct^\mr}$.  
The remaining signals required include  
$\{p_{i,t}, \dot{\theta}_{i,t} \}_{i \in \cn_{\mgfm}, t \in \ct^\mr}, \, 
\{\delta_{i,t}, \textstyle\int \delta_{i,t}, \dot{\theta}_{i,t}, \ddot{\omega}_{i,t} \}_{i \in \cn_{\mgfl}, t \in \ct^\mr}$.  
Many commercial PMUs provide customizable input and output interfaces, allowing users to access additional internal signals~\cite{SEL-2240}.  
Leveraging this capability, we utilize PMUs' internal high-frequency data over the sampling set $\ct^\ms$ to compute the remaining signals listed above: 
\begin{subequations}\label{equ:num-cal-library}
    \begin{flalign}
        & p_{i,t}=3V_{i,t}I_{i,t}\cos(\theta_{i,t}-\varphi_{i,t}), &&i\in\cn_\mgfm, t\in\ct^\ms,  \label{equ:lib-p}\\
        & \delta_{i,t} = \mathbf{T}_{q}({\theta}_{i,t})\cdot \bm{v}^{abc}_{i,t}, &&i\in\cn_\mgfl, t\in\ct^\ms,  \label{equ:lib-delta}\\
        &{\textstyle \int} \delta_{i,t} \approx \Delta T^\ms\sum_{\tau\in\ct^\ms, \tau\le t}\delta_{i,\tau}, &&i\in\cn_\mgfl, t\in\ct^\ms, \label{equ:lib-delta-int}\\
        &\dot{\theta}_{i,t} \! \approx\! ({\theta}_{i,t+\Delta T^\ms}\!-\!{\theta}_{i,t})/\!{\Delta T^\ms}, 
    &&i\in\cn, t\in\ct^\ms, \\
    &\ddot{\omega}_{i,t}\! \approx\! (\dot{\omega}_{i,t+\Delta T^\ms}\!-\!\dot{\omega}_{i,t})/\!{\Delta T^\ms},  
    && i\in\cn_{\mgfl}, t\in\ct^\ms.
    \end{flalign}
\end{subequations}
We then downsample these computed signals to the reporting set $\ct^\mr$ to align with the PMU reporting rate.

With the data stream $\{\btheta(\bx_t, \bu_t), \bdx_t\}_{t \in \ct^\mr}$ provided by PMUs, we construct the data matrices over $\ct^\mr_\mtrn \ceq \ct^\mr \cap \ct_\mtrn$ as follows: \begin{subequations}\label{equ:bTheta-bdX}
    \begin{flalign}
            &\bTheta\ceq \Big[\btheta(\bx_{t_1}, \bu_{t_1}), \cdots, \hspace*{-10cm}&& \notag \\
            & &&\btheta(\bx_{t_{|\ct^\mr_\mtrn|}},  \bu_{t_{|\ct^\mr_\mtrn|}})\Big]^{\top} \in \mathbb{R}^{{|\ct^\mr_\mtrn|}\times n_{\btheta}}, \\
        &\bdX\ceq \left[\bdx_{t_1}, \bdx_{t_2}, \cdots, \bdx_{t_{|\ct^\mr_\mtrn|}}\right]^\top \in \mathbb{R}^{{|\ct^\mr_\mtrn|}\times n_{\bx}} . \hspace*{-10cm}&&
    \end{flalign}
\end{subequations}
Accordingly, the dynamical model in vector form \eqref{equ:sindy-vector} can be reformulated in matrix form as: 
\begin{equation} \label{equ:sindy-matrix}
    \dot{\bm{X}}=\bTheta\cdot \Xi .
\end{equation}

\subsection{{The regression algorithm PC-SINDy}}\label{sec:method-regression}
Our goal is to estimate the coefficient matrix $\Xi$ (with true converter parameters) from $(\bTheta,\bdX)$. 
The original SINDy algorithm~\cite{brunton2016a} identifies the governing equations of nonlinear dynamical systems by promoting sparsity, but it is sensitive to noise in state derivatives. 
In addition, when constraints on $\omega_i$, $\omega^s_i$, and $p^s_i$ are activated, the system dynamics no longer follow \eqref{equ:sindy-matrix} as discussed in Remark~\ref{remark:constraint}. 
To address these challenges, we enhance the original SINDy algorithm as follows:  
\begin{enumerate}
    \item We replace the ordinary least squares (OLS) regression in the original SINDy~\cite{brunton2016a} with TLS regression~\cite{vanHuffel1991}, which accounts for noise in both states and derivatives. This mitigates the regression dilution effect and improves the accuracy of the estimated coefficient matrix.  
    \item We integrate RANSAC~\cite{fischler1981} to automatically reject data samples affected by constraint activation.
\end{enumerate}
Our enhanced algorithm PC-SINDy is noise-robust and constraint-adaptive; its pseudocode is given in Algorithm~\ref{alg:ca-sindy}.
\begin{algorithm}[htbp]
\caption{{Noise-robust and constraint-adaptive PC-SINDy, $\hat{\Xi}=\operatorname{PC-SINDy}(\bTheta,\bdX)$}}\label{alg:ca-sindy}
\begin{algorithmic}[1]
\Inputs $(\bTheta, \bdX)$ \color{black} under small perturbations \eqref{equ:excitation}\color{black}, maximum iterations $K_{\mathrm{RANSAC}}$, minimal subset size $m$, inlier threshold $\tau$
\State $q^\star \gets -\infty$;\quad $\mathcal{I}^\star \gets \varnothing$;\quad $\hat{\Xi}^\star \gets \mathbf{0}_{n_{\btheta}\times n_{\bx}}$
\For{$k = 1$ to $K_{\mathrm{RANSAC}}$}
  \State Randomly draw a subset \(\mathcal{J} \subseteq \ct^\mr_\mtrn\) with \(|\mathcal{J}| = m\)
  \StateH{$\hat{\Xi} \gets \operatorname{RSINDy}\!\big(\Theta[\mathcal{J},:],\, \dot{\bX}[\mathcal{J},:]\big)$ \Comment{\footnotesize Algorithm~\ref{alg:sindy}}}
  \State $\br \gets \left(\left\|\bdX[t,:] - \bTheta[t,:]\hat{\Xi}\right\|_2\,\right)_{t \in \ct^\mr_\mtrn}$
  \State $\mathcal{I} \gets \{t \in \ct^\mr_\mtrn \mid \br[t]\le \tau \}; \quad q \gets |\mathcal{I}|$ 
  \If{${q} > q^\star$}
     \State $q^\star \gets q$;\quad $\mathcal{I}^\star \gets \mathcal{I}$
  \EndIf
\EndFor
\StateH{$\hat{\Xi} \gets \operatorname{RSINDy}(\Theta[\mathcal{I}^\star,:],\, \dot{\bX}[\mathcal{I}^\star,:])$ \Comment{\footnotesize Algorithm~\ref{alg:sindy}}}
\Output $\hat{\Xi}$
\end{algorithmic}
\end{algorithm}
\begin{algorithm}
\caption{Noise-robust RSINDy, $\hat{\Xi}=\operatorname{RSINDy }(\Theta,\dot{\bX})$ }
\label{alg:sindy}
\begin{algorithmic}[1]
\Inputs $(\bTheta, \bdX)$, iterations $K_{\mathrm{SINDy}}$, sparsity threshold $\lambda$
\State $\hat{\Xi} \gets \mathbf{0}_{n_{\btheta}\times n_{\bx}}$  
\For{$i = 1$ to $n_{\bx}$}
    \State $\mathcal{S}_0 \gets \bbn_{1:n_{\btheta}}, \, \mathcal{S} \gets \mathcal{S}_0$
    \For{$k = 1$ to $K_{\mathrm{SINDy}}$}
        \StateH{$\hat{\Xi}[\mathcal{S},\, i] \gets \operatorname{TLS}\big( \bTheta[:,\,\mathcal{S}],\ \bdX[:,\,i] \big)$  \Comment{TLS~\cite{vanHuffel1991}} \label{state:ols}} 
        \State $\mathcal{S} \gets \{s \in \mathcal{S} \mid |\hat{\Xi}[s,i]| \ge \lambda\}$
        \If{$\mathcal{S} $ converged} 
            \State \textbf{break}
        \EndIf
        \State $\hat{\Xi}[\mathcal{S}_0\setminus \mathcal{S},i]\gets \bm{0}$
    \EndFor
\EndFor
\Output $\hat{\Xi}$
\end{algorithmic}
\end{algorithm}

{The outer loop of PC-SINDy (Algorithm~\ref{alg:ca-sindy}) aims to use RANSAC to identify a subset of $(\bTheta, \bdX)$ that produces the most inliers.}
As shown in Algorithm~\ref{alg:ca-sindy}, RANSAC repeatedly: (i) draws a minimal subset of time indices $\mathcal{J}\subseteq \ct^\mr_\mtrn$ with $|\mathcal{J}|=m$, (ii) calls the middle loop to estimate a provisional $\hat{\Xi}$ on $(\bTheta[\mathcal{J},:],\bdX[\mathcal{J},:])$, (iii) evaluates residuals $\br = \left(\|\bdX[t,:]-\bTheta[t,:]\hat{\Xi}\|_2\,\right)_{t\in\ct^\mr_\mtrn}$,
and (iv) forms the inlier set $\mathcal{I}=\{t\in\ct^\mr_\mtrn\mid r_t\le \tau\}$ with score $q=|\mathcal{I}|$. The iteration that yields the largest inlier set $\mathcal{I}^\star$ is retained, and the final model $\hat{\Xi}$ is obtained by fitting on $\mathcal{I}^\star$. 
In our implementation, the minimal subset size \(m\) is chosen as \(m = 3n_{\btheta}\) to improve numerical stability while maintaining computational efficiency. 

The inner loop $\hat{\Xi} = \operatorname{RSINDy}(\cdot, \cdot)$ adopts the original SINDy algorithm (Code~1 in~\cite{brunton2016a}) with the OLS step in line~\ref{state:ols} replaced by TLS~\cite{vanHuffel1991}, as detailed in Algorithm~\ref{alg:sindy}. For each state $i \in \{1, \dots, n_{\bx}\}$, the algorithm iteratively: (i) performs regression over the active set $\mathcal{S}$, (ii) prunes $\mathcal{S}$ by thresholding $|\hat{\Xi}[s,i]| < \lambda$, and (iii) sets $\hat{\Xi}[s,i] = 0$ for all $s \notin \mathcal{S}$. This loop repeats until convergence or reaching the maximum iteration count $K_{\mathrm{SINDy}}$.

The proposed identification algorithm $\hat{\Xi}= \operatorname{PC-SINDy}(\bTheta,\bdX)$ provides a probabilistic guarantee of convergence as presented in the following proposition. 
{\begin{proposition}\label{proposition:1}
Let $w$ denote the proportion of inliers in the data, $m$ the minimal subset size, and $p$ the desired probability of convergence. Then, in the worst case, the minimal number of iterations required to ensure Algorithm~\ref{alg:sindy} converges to a local minimizer with probability at least $p$ is given by 
    \begin{equation}\label{equ:convergence}
        K\ceq n_{\bx} n_{\btheta} \left\lceil {\log(1-p)}/{\log(1-w^m)} \right\rceil, 
    \end{equation}
    where $\lceil \cdot \rceil$ denotes the ceiling function. 
\end{proposition}}
{\begin{proof}
The ensemble algorithm $\hat{\Xi}= \operatorname{PC-SINDy}(\bTheta,\bdX)$ consists of an outer RANSAC loop and an inner RSINDy loop. For the outer loop, the minimal number of RANSAC iterations required to ensure convergence with probability at least $p$ is $K_{\mathrm{RANSAC}} \ceq \left\lceil {\log(1-p)}/{\log(1-w^m)} \right\rceil$ (see Sec.~II\text{-}B of \cite{fischler1981}). 
The inner RSINDy loop guarantees convergence {to a local minimizer} within $K_{\mathrm{SINDy}} \ceq n_{\btheta}$ iterations for each state $i \in \bbn_{1:n_{\bx}}$~\cite{brunton2016a}. 
Therefore, in the worst case, the total number of TLS regressions is
\begin{equation}\label{equ:total-iteration-number}
    K \ceq n_{\bx} K_{\mathrm{RANSAC}} K_{\mathrm{SINDy}},
\end{equation}
and the algorithm returns a consistent estimate $\hat{\Xi}$ with probability at least $p$.
Hence, the iteration bound~\eqref{equ:convergence} holds.
\end{proof}}

Probabilistic convergence in Proposition~\ref{proposition:1} refers to \emph{inlier-consistent identification}: PC-SINDy (Algorithm~\ref{alg:ca-sindy}) is guaranteed, with a prescribed probability, to find a coefficient matrix $\hat{\Xi}$ consistent with the inlier data.  %
Proposition~\ref{proposition:1} therefore characterizes the computational efficiency of the proposed PC-SINDy.

For an MG with $N_{\mathrm{DER}}$ DERs, we have the number of state variables $n_{\bx}\leq 3N_{\mathrm{DER}}$, the dimension of the library $n_{\btheta} = 4N_{\mathrm{DER}}$, and the minimal subset size $m = 3n_{\btheta} = 12N_{\mathrm{DER}}$.
Under the small perturbations \eqref{equ:excitation} and the constraints\footnote{The frequency limits for GFM and GFL converters are [57~Hz, 61.8~Hz] and [48~Hz, 66~Hz], respectively, with corresponding ROCOF limits of at least $\pm 5$~Hz/s and $\pm 10$~Hz/s.} specified in IEEE Std~C37.118.1-2011~\cite{IEEE2011}, the inlier proportion $w$ exceeds {$99.9\%$}, and we simply set {$w = 0.999$}. 
Even with a high success probability $p = 99.9\%$, the worst-case iteration count satisfies
\begin{flalign}
    K &\leq 12N_{\mathrm{DER}}^2 \underbrace{\left\lceil {\log(1-p)}/{\log(1-w^{12N_{\mathrm{DER}}})} \right\rceil}_{\approx N_{\mathrm{DER}}/4, \; \forall N_{\mathrm{DER}}\leq 20} \nonumber \hspace*{-10em}\\
      &\leq 12N_{\mathrm{DER}}^2 \times {N_{\mathrm{DER}}}/{3} \nonumber \\
      &\leq 4N_{\mathrm{DER}}^3, && \forall N_{\mathrm{DER}}\leq 20,\label{equ:remark-cubical}
\end{flalign}
i.e., the worst-case iteration count scales as $\mathcal{O}(N_{\mathrm{DER}}^3)$. Although the complexity grows cubically, each iteration ($\hat{\Xi}[\mathcal{S},\, i] \gets \operatorname{TLS}\big( \bTheta[:,\,\mathcal{S}],\ \bdX[:,\,i] \big)$, line \ref{state:ols} of Algorithm \ref{alg:sindy}) is computationally inexpensive. Consequently, even for a large MG with $N_{\mathrm{DER}} = 20$ DERs, the computation time under the worst-case iteration counts remains below $2$~s, as measured in \textsc{MATLAB}~R2024a on an Intel Core~i7-7700 CPU at 3.60~GHz. %

For the 4-DER MG considered in Section~\ref{sec:simulation}, which includes one GFM and three GFLs, we have $n_{\bx}=11$, $K_{\mathrm{SINDy}} = n_{\btheta} = 16$, and $m = 3n_{\btheta} = 36$. In this case, PC-SINDY requires at most $K={528}$ total iterations to converge in the worst case according to \eqref{equ:total-iteration-number}. In practice, however, it converges in only $126$ iterations and takes $8.6$~ms under the aforementioned computational environment, which demonstrates its high computational efficiency.

\subsection{Frequency {Control} with PC-SINDYc}\label{sec:method-mpc}
Once the coefficient matrix $\hat{\Xi}=\operatorname{PC-SINDy}(\bTheta,\bdX)$ is identified offline via {Algorithm \ref{alg:ca-sindy}} under small perturbations, the obtained $\hat{\Xi}$ can be subsequently employed to design a controller for online frequency control under various disturbances. 
As discussed in Section~\ref{sec:DER-dynamics}, since GFM converters regulate the system frequency via droop or VSG control and GFL converters passively follow it through PLLs, only GFMs participate in frequency regulation. We therefore focus on the GFM frequency dynamics and introduce the following frequency-related quantities: 
\begin{flalign}
    &\tx_i \cceq \bx_i[2]=\omega_i, \, 
    \tcx_i \cceq \{\tx_i \mid \exists \bx_i[1], \bx_i\in\cx_i\},  \,
    \tXi_i \cceq \Xi_i[:,2], \, \hspace*{-10em} && \nonumber\\
    &\hat{\tXi}_i \cceq \hat{\Xi}_i[:,2], 
    && i\in\cn_\mgfm,
\end{flalign}
where $\tx_i=\bx_i[2]$ is the second entry of state vector $\bx_i$, $\hat{\Xi}_i$ is the $i$th block of the estimated matrix $\hat{\Xi}$. 
We define the stacked variables as
\begin{subequations}
\begin{flalign}
    &\tbx \!\cceq\! (\tx_i)_{i\in\cn_\mgfm}, \hspace*{-10em}
    && \tbu \!\cceq\! (\bu_i)_{i\in\cn_\mgfm}, \\
    &\tbtheta(\tbx,\! \tbu) \!\!\cceq\! (\btheta_i(\bx_i,\! \bu_i))_{i\in\cn_\mgfm}, \hspace*{-10em}
    && \tXi \!\cceq\! \diag((\tXi_i)_{i\in\cn_\mgfm}), \label{equ:stack-mpc-2}\\
    &\tcx \!\cceq\! \textstyle\prod_{i\in\cn_\mgfm}\tcx_i, \hspace*{-10em}
    && \tcu \!\cceq\! \textstyle\prod_{i\in\cn_\mgfm}\cu_i,
\end{flalign}
\end{subequations}
and the discrete-time frequency $\tbx^+$ dynamics of GFM converters as
\begin{equation}\label{equ:sindy-mpc-discrete}
    \tbx^{+} = \tbF(\tbx,\tbu)=\tbx + \Delta T^{\mr} \hat{\tXi}^{\top}\tbtheta(\tbx,\tbu).
\end{equation}

\subsubsection{MPC}
At each instant $t \in \ct^\mr \setminus \ct^\mr_{\mathrm{trn}}$, we solve the optimal control problem (OCP) defined below over a prediction horizon of $N$ steps: 
\begin{subequations}\label{equ:ocp}
\begin{flalign}
\text{Obj.}\ 
&\min_{\left(\tbu(k)\right)_{k=0}^{N-1}} 
\sum_{k=0}^{N-1}\ell(\tbx(k), \tbu(k))+V_f(\tbx(N)) \hspace*{-10em}&& \label{equ:ocp-obj}\\
\text{s.t.}\ 
& \tbx(0)=\tbx_t, \tbx(N)\in\tcx_f, && \label{equ:ocp-constraint}\\
& \tbx(k+1) = \tbF(\tbx(k),\tbu(k)), &&k\in \bbn_{0:N-1}, \label{equ:ocp-iteration}\\ 
& \tbx(k)\in \tcx,  \tbu(k)\in\tcu,  \hspace*{-10em} &&k\in \bbn_{0:N-1}. \label{equ:ocp-all-constraint}
\end{flalign}
\end{subequations}
{Solving OCP~\eqref{equ:ocp} yields the receding-horizon MPC
\begin{equation}\label{equ:mpc-controller}
    \bkappa(\tbx) \ceq \tbu^\star(0), \quad \tbu=\bkappa(\tbx),
\end{equation}
which applies the first element $\tbu^\star(0)$ of the optimal control sequence $\left(\tbu^\star(k)\right)_{k=0}^{N-1}$ and shifts the horizon forward.
Here, $\tbx(k)$ and $\tbu(k)$ denote the predicted state and control decision, respectively, at step $k$ (i.e., at time $t + k\Delta T^\mr$), with initial condition $\tbx(0) = \tbx_t$. 
The terminal set $\tcx_f$, stage cost function $\ell(\cdot)$, and terminal cost function $V_f(\cdot)$ are \emph{to be determined}. }

{In the following subsections, we will design suitable $\tcx_f$, $\ell$, and $V_f$ such that the resulting control law $\bkappa(\cdot)$ renders the closed-loop system $\tbx^{+} = \tbF(\tbx, \bkappa(\tbx))$ asymptotically stable with respect to the reference state $\tbx^{\mref}$. }

\subsubsection{Stage cost $\ell(\cdot)$}
Let the stage cost $\ell(\cdot)$ be defined as
\begin{align}
    & \ell(\tbx,\tbu)\ceq \left\| \tbx - \tbx^{\mref} \right\|_{\bm{Q}}^2 + \| \Delta\tbu \|_{\bm{R}}^2, \label{equ:ell}
\end{align}
where 
\begin{align}
    & \Delta \tbu(k)\ceq 
    \begin{cases}
        \bm{0}, & k=0,\\
        \tbu(k)-\tbu(k-1), & k\in \bbn_{1:N-1}.
    \end{cases} 
\end{align}
The reference state vector $\tbx^\mref$ and stage cost weights $\bQ \succeq 0$ and $\bR \succ 0$ are defined as 
\begin{align}\label{equ:weighting-matrices}
& \tbx^{\mref} \ceq \omega_0\bm{1}_{n_{\tbx}}, \quad 
\bQ \ceq q\bfi_{n_{\tbx}}, \quad 
\bR \ceq r\bfi_{n_{\tbu}},
\end{align}
where $\bfi_n$ denotes the $n\times n$ identity matrix. 
All GFM converters share the same reference frequency $\omega_0$ and weighting coefficients $q$ and $r$, {where a larger $q$ leads to more aggressive frequency correction, whereas a larger $r$ results in smoother control actions. 
In our implementation, we set $q=r=1$.} 

\subsubsection{Terminal set $\tcx_f$ and terminal cost $V_f(\cdot)$}
Throughout this subsection, let $i\in\cn_{\mgfm}$.
For each GFM $i$, the stacked model in~\eqref{equ:sindy-mpc-discrete} can be expressed in component form as
\begin{equation}
    \omega^+_i=\tbF_i(\tx_i,\bu_i)=\omega_i+\Delta T^\mr\hat{\omega}^c_i(\omega^s_i-\omega_i+\hat{\sigma}_i(p^s_i-p_i)), 
\end{equation}
where $\hat{\omega}^c_i \ceq \hat{\Xi}_i[2,2]$ and $\hat{\sigma}_i \ceq \hat{\Xi}_i[4,2]/\hat{\Xi}_i[2,2]$.

Let the frequency error be $e_i \ceq \omega_i-\omega_0$ and consider a terminal control law
\begin{align}\label{equ:kappa_f_i-gfm}
    & \bkappa_{f,i}(\tx_i) \ceq 
    \begin{bmatrix}\omega^s_{i}\\ p^s_{i}\end{bmatrix} = 
    \begin{bmatrix}\omega_0 - \alpha_i e_i \\ p_i - \beta_i e_i\end{bmatrix}.
\end{align}
Then the closed-loop error dynamics are
\begin{align}\label{equ:e-dym}
    & e^+_i=\phi_i e_i, \quad   \phi_i \ceq 1-\Delta T^\mr \hat{\omega}^c_i(1+\alpha_i +\hat{\sigma}_i\beta_i). 
\end{align}
Specifically, $\alpha_i,\beta_i>0$ in \eqref{equ:kappa_f_i-gfm} are selected such that $|\phi_i|<1$, which is easy to satisfy given the fact that $\Delta T^\mr$ is a known positive value while $\hat{\omega}^c_i>0$ and $\hat{\sigma}_i>0$ are known from the estimated coefficient matrix $\hat{\Xi}$. 
We define
\begin{subequations}\label{equ:Vf-Xf}
    \begin{align}
        &\bkappa_f(\tbx)\ceq(\bkappa_{f,i}(\tx_i))_{i\in\cn_\mgfm},\\
        &\tcx_f \ceq \textstyle\prod_{i\in\cn_\mgfm} \tcx_{f,i}, \label{equ:cXf}\\
        &V_f(\tbx) \ceq \left\| \tbx - \tbx^{\mref} \right\|_{\bP_f}^2,    
    \end{align}
\end{subequations}
where 
\begin{subequations}\label{equ:gfm-Xf} 
\begin{flalign}
    &\tcx_{f,i} \ceq \{\tx_i \in \tcx_i \mid |e_i|\leq {\varepsilon}, \, \bkappa_{f,i}(\tx_i)\in\cu_i\}, \hspace*{-20em}&&\label{equ:cxfi}\\
    &\bP_f \ceq \diag(p_{f,i})_{i\in\cn_\mgfm},&&\\
    &p_{f,i} \ceq {(1+(\alpha_i^2+\beta_i^2){(1-\phi_i)^2/\phi_i^2})}/{(1-\phi_i^2)}. \hspace*{-10em}&& \label{equ:pf}
\end{flalign}
\end{subequations}
{Here, the small positive threshold $\varepsilon>0$ controls the size of $\tcx_{f,i}$. In our implementation, we set $\varepsilon=0.003$ p.u., corresponding to 0.18 Hz in a 60 Hz system.}

\begin{lemma}[Theorem 2.19(b) of \cite{rawlings2017}]\label{lemma:rpi}
The reference state $\bx^\mref$ is asymptotically stable for the closed loop $\bx^{+}=\bF\big(\bx,\bkappa(\bx)\big)$ provided the following conditions hold:
\begin{enumerate}[label=(\alph*)]
\item for all $\bx\in \cx_f$ there exists a terminal control law $\bkappa_f(\bx)\in\cu$ such that 
    \begin{subequations}
        \begin{flalign}
            & V_f(\bF(\bx,\bkappa_f(\bx)))-V_f(\bx)\leq -\ell(\bx,\bkappa_f(\bx)), \hspace*{-10em}&& \label{equ:condition1-Vf}\\
            & \bF(\bx,\bkappa_f(\bx))\in\cx_f\subseteq\cx, && \label{equ:condition1-rpi}
        \end{flalign}
    \end{subequations} \label{lemma:c1}
\item there exist $\mathcal{K}_\infty$ functions $\alpha_1(\cdot)$ and $\alpha_f(\cdot)$ satisfying 
    \begin{subequations}
        \begin{flalign}
            & \ell(\bx,\bu)\geq \alpha_1(|\bx-\bx^{\mref}|),\hspace*{-10em} && \forall (\bx,\bu)\!\in\! \cx\!\times\!\cu, \label{equ:condition2-ell}\\
            & V_f(\bx)\leq \alpha_f(|\bx-\bx^{\mref}|), \hspace*{-10em} && \forall \bx\in\cx_f, \label{equ:condition2-Vf} 
        \end{flalign}
    \end{subequations} \label{lemma:c2}
\end{enumerate}
\end{lemma}
\vspace{-1.5em}
Lemma~\ref{lemma:rpi} provides sufficient conditions for asymptotic stability of a general closed-loop system $\bx^{+} = \bF\big(\bx,\bkappa(\bx)\big)$. By applying this lemma to our closed-loop MG system, we obtain Theorem~\ref{theorem:1}.
\begin{theorem}\label{theorem:1}
Let $\bkappa_f$, $\tcx_f$, and $V_f$ be defined as in~\eqref{equ:Vf-Xf}. 
{If the OCP~\eqref{equ:ocp} is feasible at the initial time $t_0 \in \ct^\mr \setminus \ct^\mr_\mtrn$, then for all $t \in \ct^\mr \setminus \ct^\mr_\mtrn$ with $t \ge t_0$, OCP~\eqref{equ:ocp} remains feasible, and} the control law $\bkappa(\cdot)$ obtained by solving~\eqref{equ:ocp} renders the closed-loop system $\tbx^{+} = \tbF(\tbx, \bkappa(\tbx))$ asymptotically stable with respect to the reference state $\tbx^{\mref}$.
\end{theorem}
\begin{proof}
Throughout this proof, let $i\in\cn_{\mgfm}$.
We will prove that conditions \ref{lemma:c1}–\ref{lemma:c2} of Lemma~\ref{lemma:rpi} hold. Then, we invoke Lemma~\ref{lemma:rpi} to conclude the theorem.

\textit{1) Condition \ref{lemma:c1}:}
Consider the component-wise cost functions
\begin{subequations}
\begin{align}
  & \ell_i(\tx_i,\bu_i)\ceq qe_i^2 + \| \Delta \bu_i \|^2_{\diag(r,r)}, \\
  & V_{f,i}(\tx_i) \ceq p_{f,i} e_i^2. 
\end{align}
\end{subequations}
where 
\begin{align}\label{equ:delta-ui}
    & \Delta \bu_i(k)\ceq 
    \begin{cases}
        \bm{0}, & k=0,\\
        \bu_i(k)-\bu_i(k-1), & k\in \bbn_{1:N-1}.
    \end{cases} 
\end{align}
For all $\tx_i\in\tcx_{f,i}$, there exists a terminal control law $\bkappa_{f,i}$ defined in~\eqref{equ:kappa_f_i-gfm} such that 
\begin{subequations}
\begin{flalign}
    &V_{f,i}(\tbF_i(\tx_i,\bkappa_{f,i}(\tx_i)))-V_{f,i}(\tx_i) \overset{\eqref{equ:e-dym}}{=}p_{f,i}e_i^2(\phi^2_i-1) &&\nonumber\\
    &\hspace{1em}\overset{\eqref{equ:pf}}{=}-{(1+(\alpha_i^2+\beta_i^2){(1-\phi_i)^2}/\phi_i^2)}e_i^2, \label{equ:c1-Vfi} &&\\
    &\ell_{i}(\tx_i,\bkappa_{f,i}(\tx_i)) \overset{\eqref{equ:delta-ui}}{\ge} qe_i^2+\|\bkappa_{f,i}(\tx_i)-\bkappa_{f,i}(\tx_i^-)\|^2_{\diag(r,r)} \nonumber&&\\
    &\hspace{1em}= qe_i^2+
    \left\|\begin{bmatrix} \omega_0-\alpha_i e_i\\p_i-\beta_i e_i \end{bmatrix} - 
    \begin{bmatrix} \omega_0-\alpha_i e_i^-\\p_i-\beta_i e_i^- \end{bmatrix}\right\|^2_{\diag(r,r)} \nonumber&&\\
    &\hspace{1em}\overset{\eqref{equ:e-dym}}{=}qe_i^2+
    \left\|\begin{bmatrix} \alpha_i e_i({1}/{\phi_i}-1)\\\beta_i e_i({1}/{\phi_i}-1) \end{bmatrix}\right\|^2_{\diag(r,r)} \nonumber&&\\
    &\hspace{1em}={(1+(\alpha_i^2+\beta_i^2){(1-\phi_i)^2}/\phi_i^2)}e_i^2. \label{equ:c1-elli}&&
\end{flalign}
\end{subequations}

Combining~\eqref{equ:c1-Vfi} and~\eqref{equ:c1-elli} yields
\begin{equation}\label{equ:condition1-Vfi}
    V_{f,i}(\tbF_i(\tx_i,\bkappa_{f,i}(\tx_i)))-V_{f,i}(\tx_i)
  \le-\ell_i(\tx_i,\bkappa_{f,i}(\tx_i)). 
\end{equation}
Summing~\eqref{equ:condition1-Vfi} over all $i\in\cn_\mgfm$ yields the stacked relation~\eqref{equ:condition1-Vf}.

For all $\tx_i\in\tcx_{f,i}$, we have
\begin{align}
    &|e_i^+|=|\tbF_i(\tx_i,\bkappa_{f,i}(\tx_i))-\omega_0|=|\phi_ie_i| \overset{|\phi_i|<1}{<}|e_i|.
\end{align}
Hence, the successor state $\tbF_i(\tx_i,\bkappa_{f,i}(\tx_i))$ satisfies the terminal constraint \eqref{equ:cxfi}, i.e.,
\begin{align}
    &\tbF_i(\tx_i,\bkappa_{f,i}(\tx_i))\in\tcx_{f,i},  \nonumber\\
    \Rightarrow& \left(\tbF_i(\tx_i,\bkappa_{f,i}(\tx_i))\right)_{i\in\cn_\mgfm}\in \textstyle\prod_{i\in\cn_\mgfm}\tcx_{f,i},  \nonumber\\
    \Rightarrow& \tbF(\tbx,\bkappa_{f}(\tbx))\in\tcx_{f} \subseteq \tcx,
\end{align}
thereby~\eqref{equ:condition1-rpi} holds.
As such, condition \ref{lemma:c1} of Lemma~\ref{lemma:rpi} is satisfied. 

\textit{2) Condition \ref{lemma:c2}:}
Let 
\begin{subequations}
\begin{flalign}
&\be\ceq\tbx-\tbx^\mref, &&\\
&\alpha_1\!\left(|\be|\right) \ceq \|\be\|_{\bQ/2}^2 \in \mathcal{K}_\infty, \;\;
\alpha_f\!\left(|\be|\right) \ceq \|\be\|_{2\bP_f}^2 \in \mathcal{K}_\infty .\hspace*{-10em}&&
\end{flalign}
\end{subequations}
Since $\ell(\tbx,\tbu)$ and $V_f(\tbx)$ are quadratic with $\bQ\succeq0$, $\bR\succ0$, and $\bP_f\succeq0$, we have
\begin{subequations} 
\begin{align} 
&\ell(\tbx,\tbu)>\|\be\|_{\bQ}^2\ge\|\be\|_{\bQ/2}^2 =\alpha_1\big(|\be|\big),\\ 
& V_f(\tbx)=\|\be\|_{\bP_f}^2\le\|\be\|_{2\bP_f}^2 =\alpha_f\big(|\be|\big). 
\end{align} 
\end{subequations}
Hence condition~\ref{lemma:c2} of Lemma~\ref{lemma:rpi} is satisfied.

Here we have verified conditions \ref{lemma:c1}–\ref{lemma:c2} of Lemma~\ref{lemma:rpi}, and thus conclude the theorem by invoking the lemma.
\end{proof}

\begin{remark}\label{remark:mpc-complexity}
The computational cost of solving OCP \eqref{equ:ocp} increases with both the prediction horizon $N$ and the number of GFMs $N_{\mgfm}$. Since PC-SINDy linearizes the MG frequency dynamics, OCP \eqref{equ:ocp} reduces to a quadratic program (QP) at each time step \cite{jerez2012}, with complexity $\mathcal{O}(N^3 N_{\mgfm}^3)$. By exploiting the block-diagonal structure $\tXi \ceq \diag((\tXi_i)_{i\in\cn_\mgfm})$ in \eqref{equ:stack-mpc-2}, OCP \eqref{equ:ocp} decouples across GFMs and can be solved independently for each GFM $i$. The complexity then reduces to $\mathcal{O}(N^3 N_{\mgfm})$, i.e., cubic in $N$ but linear in $N_{\mgfm}$. Even for a large MG with $N_{\mgfm}=20$, the computation time per control step is about $6$~ms, which is below $\Delta T^\mr=1/120$~s, confirming the real-time feasibility of PC-SINDYc.
\end{remark}

\begin{remark} \label{remark:N-selection}
Typically, the feasibility of OCP ~\eqref{equ:ocp} at the initial time $t_0$ is easy to satisfy. Even if OCP~\eqref{equ:ocp} is infeasible at the initial time $t_0$, feasibility can be restored by %
increasing $N$ or enlarging the terminal set $\tcx_f$, e.g., increasing $\varepsilon$ or adjusting $\alpha_i$ and $\beta_i$ in $\bkappa_{f,i}(\tx_i)$ for $i\in\cn_\mgfm$ (see \eqref{equ:kappa_f_i-gfm} and \eqref{equ:cxfi}).

In our simulations, $N = 5$ is sufficient to handle all large disturbances %
due to the moderately large frequency-regulation deadband $\varepsilon = 0.003\text{ p.u.}$ $( 0.18\text{ Hz})$, which lies well within the range $0.015\text{ Hz} \leq \varepsilon \leq 0.96\text{ Hz}$ suggested by IEEE Std 2800-2022~\cite{IEEE2022}, Table~7. This relatively large terminal set %
allows the controller to drive the state into $\tcx_f$ within two steps in our simulations, indicating $N=5$ is adequate. %
Also, $N$ cannot be too large, because the computational complexity scales as $\mathcal{O}(N^3)$ as discussed in Remark~\ref{remark:mpc-complexity}. Therefore, $N=5$ represents a good trade-off between control performance and
computational cost. 

\end{remark}

\begin{remark}\label{remark:theorem-limitation}
The asymptotic stability established in Theorem~\ref{theorem:1} applies to the identified discrete-time model \eqref{equ:sindy-mpc-discrete}, rather than to the real-world MG frequency dynamics. %
Mismatch between the identified \eqref{equ:sindy-mpc-discrete} and real-world MG dynamic model may arise from PMU noise,  %
errors between $\Xi$ and $\hXi$, and discretization in \eqref{equ:sindy-mpc-discrete}. Nevertheless, Section~\ref{sec:sim-xi} shows that the proposed PC-SINDy accurately estimates $\hXi$ despite PMU noise, and Section~\ref{sec:sim-prediction} demonstrates accurate prediction of future states. Consequently, the proposed PC-SINDYc maintains good control performance and stability in practice (Sections~\ref{sec:sim-control}--\ref{sec:sim-multiple-ders}). A rigorous robustness guarantee under model mismatch would require more advanced tools, such as tube-based MPC that is planned in our future work.

\end{remark}

\subsection{Step-by-Step Implementation}
Here is the step-by-step procedure of the PC-SINDYc framework for system identification and frequency control.
\begin{algorithm}[htbp]
\caption{PC-SINDYc framework for frequency-dynamics identification and frequency control.}
\label{alg:step-by-step}
\begin{enumerate}[label=\textbf{Step~\arabic*:},ref=\textbf{Step~\arabic*}, leftmargin=0.7em, align=left, labelsep=0.5em]
\item Generate the small-perturbation signal $\{\bu_t\}_{t\in\ct_\mtrn}$ to excite the DER-based MG during the offline training period $\ct_\mtrn$ (Section~\ref{sec:method-activation}). \label{step:1}
\item Use PMUs to collect measurements over $\ct_\mtrn$ and construct the data matrices $\bTheta$ and $\bdX$ as in~\eqref{equ:bTheta-bdX} (Section~\ref{sec:PMU}). \label{step:2}
\item Identify the \emph{true} coefficient matrix via the PC-SINDy algorithm: $\hat{\Xi}=\operatorname{PC-SINDy}(\bTheta,\bdX)$ (Section~\ref{sec:method-regression}). \label{step:3}
\item With the identified coefficient matrix $\hat{\Xi}$, set up OCP~\eqref{equ:ocp} with weighting matrices from~\eqref{equ:weighting-matrices} and terminal cost/set from~\eqref{equ:Vf-Xf}; its solution yields MPC~\eqref{equ:mpc-controller}, which guarantees asymptotic stability (Section~\ref{sec:method-mpc}).
\end{enumerate}
\end{algorithm}

Note that $\hat{\Xi}$ identified in \ref{step:3} contains the true physical system coefficients. Once identified, the SINDy structure $\dot{\bx} = \Xi^\top \btheta(\bx, \bu)$, driven by real-time measurements, can be directly used to design an MPC for online frequency control. This enables the MPC to handle large disturbances unseen during training and avoids retraining common in machine learning methods.

\section{Case Study}\label{sec:simulation}

\begin{figure}[htbp]
    \centering
    \includegraphics[width=\linewidth]{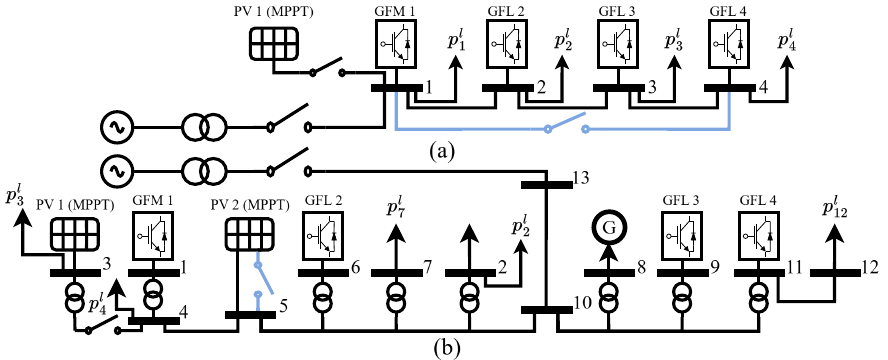}
    \caption{(a) The 4-bus MG and (b) 13-bus MG.} 
    \label{fig:MG-circuit}
\end{figure}

In this section, we apply the proposed PC-SINDYc framework to several DER-based MGs and compare it against existing methods. 
For identification, we benchmark against a SINDy-based method that employs a conventional, intuitively developed library of candidate functions \cite{nandakumar2023}; we refer to this benchmark as \emph{conventional SINDy} for brevity.
For frequency control, we compare with MPC built on the conventional SINDy model \cite{nandakumar2023}, referred to as \emph{conventional SINDYc}, a well-tuned PI controller \cite{savaghebi2012}, and state-of-the-art RL controllers \cite{rodriguez-martinez2025}.

In MATLAB/Simulink, we build two testbeds: a low-voltage 4-bus MG and a medium-voltage 13-bus MG with transformers and diesel generators to emulate more complex electromagnetic scenarios (Figs.~\ref{fig:MG-circuit}(a)–(b)). Both MGs include GFM and GFL converters, PV plants operating in maximum power point tracking (MPPT) mode, and loads $p^l_i$ at buses $i \in \cn_\mload$ that randomly oscillate around 0.5~p.u. with magnitude 0.01~p.u. (see \eqref{equ:psi}):
\begin{align}
p^l_{i,t}
&\ceq 0.5+\psi(t;0.01,50,50\pi),
&i \in \cn_\mload,\ t \in \ct_\mtrn.
\label{equ:load-excitation}
\end{align}

The converters, including their inner control loops and filters, are modeled in detail to simulate electromagnetic-transient (EMT) dynamics. The GFM converter operates in either droop or VSG mode, and the GFL converter uses either a single-phase or three-phase PLL. 
A PMU is installed at the point of common coupling of each converter.
Each MG connects to a distribution grid through a breaker and transformer, enabling grid-connected or islanded operation. The rated frequency is \(f_0=60~\text{Hz}\); the remaining MG parameters follow \cite{gong2023a}.

\subsection{System Identification by PC-SINDy Under Noise and Low-Sampling Rate}\label{sec:sim-xi}
We perform identification on the 4-bus MG with GFM~1 operating in droop mode and GFLs~2–4 using three-phase PLLs. 
In \ref{step:1}, we excite the MG with small input perturbations $\bu$ of amplitude $\pm 0.01$ p.u. over $T_\mtrn=10~\mathrm{s}$ \eqref{equ:excitation}. The input $\bu$ oscillates in the range $(-50\pi,\,50\pi)\,\mathrm{rad/s}$. The PMU reporting rate is $F^\mr=120$ Hz according to IEEE Std C37.118.1-2011~\cite{IEEE2011}. 

In \ref{step:2}, we adopt PMU accuracy consistent with IEEE Std C37.118.1-2011: angular-frequency error within $\pm 0.031~\mathrm{rad/s}$ and its derivative error within $ \pm 0.063~\mathrm{rad/s^2}$. For measurements not specified by the standard (e.g., three-phase voltages), noise levels are set to match \cite{gong2023a}. It should be noted that the library matrix $\bTheta(\bX,\bU)$ and the derivative matrix $\bdX$ are computed numerically from these \emph{noise-corrupted} measurements according to~\eqref{equ:num-cal-library}.

In \ref{step:3}, we apply $\operatorname{PC-SINDy}(\cdot)$ to estimate $\hat{\Xi}$ and compare it with the analytical matrix $\Xi$~\eqref{equ:whole-defs} (Fig.~\ref{fig:Xi-all}). 
Clearly, $\hat{\Xi}$ closely matches $\Xi$, with {a small relative Frobenius norm error $\|\hXi-\Xi\|_F/\|\Xi\|_F=1.61\%$}, indicating that we \emph{accurately} identify the \emph{true} frequency-dynamics model from noise-corrupted measurements without overfitting.  
The identification takes only $8.6$~ms in \textsc{MATLAB}~R2024a on an Intel Core i7-7700 at 3.60~GHz.

\begin{figure}[htbp]
    \centering
    \includegraphics[width=0.9\linewidth]{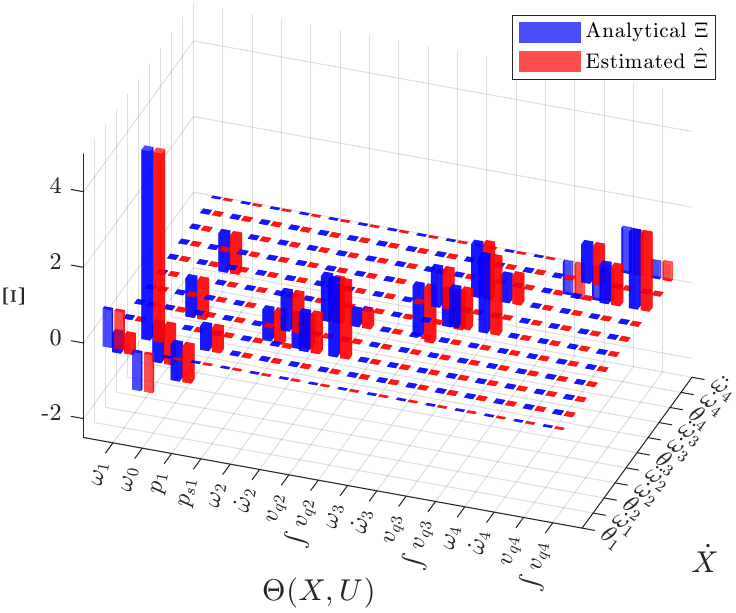}
    \caption{Comparison of true \(\Xi\) and PC-SINDy-estimated \(\hat{\Xi}\) using normalized high-quality and low-quality measurements.}
    \label{fig:Xi-all}
\end{figure}

\subsection{Frequency Prediction}\label{sec:sim-prediction}

In this subsection, we validate the prediction accuracy of the identified model under {large-signal disturbances that has never been seen during the identification process}.
Unlike the training scenario, which contains only small perturbations of $\pm 0.01$~p.u., the validation includes large disturbances unseen during training to evaluate the identified model’s generalization. 
The validation begins at $t=10$ s.
At $t=10.5$ s, the load at bus 3 steps up by $0.3$ p.u.; at $t=11$ s, PV 1, operating at 0.5 p.u. of its rated power, connects to the MG. 
During $t \in [11.5,11.8]$~s, {the perturbation amplitude of the load at bus~2} is increased by a factor of 10 (from $\pm 0.01$ p.u.\ to $\pm 0.1$ p.u.) and then returned to $\pm 0.01$ p.u.\ at $t = 11.8$~s.
At $t=12$ s, {the auxiliary line (blue in Fig.~\ref{fig:MG-circuit}) is switched in (breaker closed)}, resulting in a topology change; the MG then connects to the AC distribution grid at $t=12.5$ s and disconnects at $t=13$ s.

\begin{figure}[htbp]
    \centering
    \includegraphics[width=\linewidth]{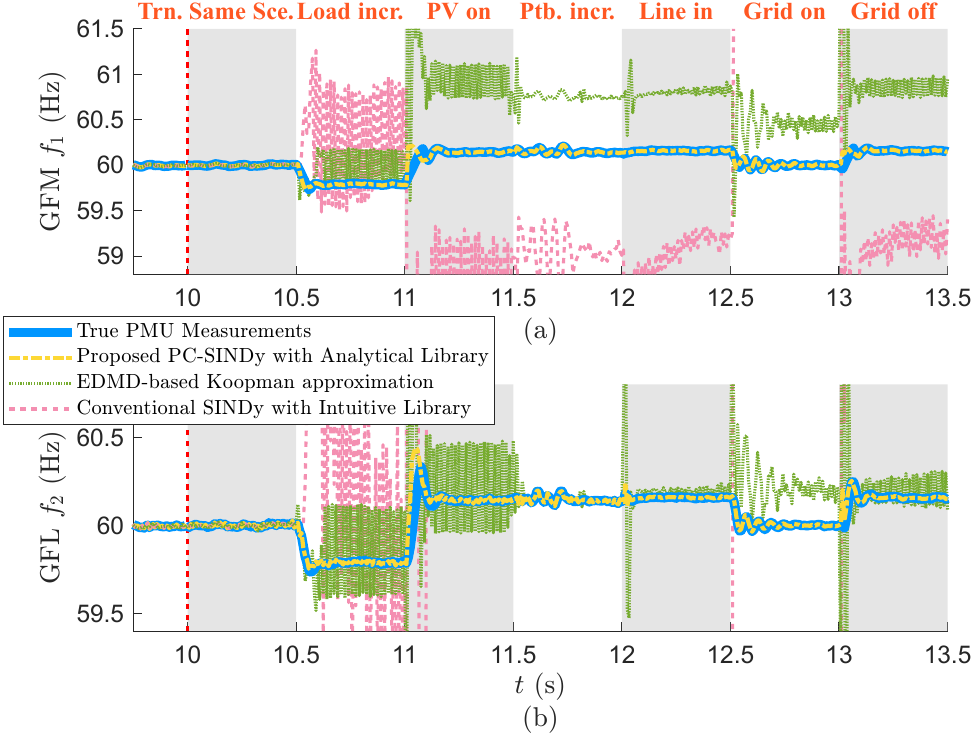}
    \caption{Predicted GFM frequency $f_1$ and GFL frequency $f_2$ using the proposed PC-SINDy model, {EDMD-based Koopman approximation~\cite{arbabi2018}, }and the conventional SINDy benchmark~\cite{nandakumar2023}.}
    \label{fig:4bus-prediction}
\end{figure}
We employ the identified coefficient matrix $\hat{\Xi}$ from \ref{step:3} to predict future states
\begin{equation}\label{equ:sindy-discrete}
    \bx^{+} = \bx + \Delta T^{\mr}\,\hat{\Xi}^{\top}\btheta(\bx,\bu),
\end{equation}
\newcommand{\mben}{\mathrm{bm}}
and compare the predictions against two benchmarks: an extended dynamic mode decomposition (EDMD)–based Koopman approximation with delay embedding~\cite{arbabi2018}, and conventional SINDy~\cite{nandakumar2023}.

The EDMD-based Koopman approximation~\cite{arbabi2018} uses the same state $\bx$ and input $\bu$ variables as the proposed PC-SINDYc framework does and includes PMU measurements into its observable vector $\bm h(\bx)$:
\begin{equation}\label{equ:bm-h}
    \bm h(\bx)=(\omega_i,~ \dot{\omega}_i,~ V_i, ~\theta_i, ~ I_i, ~ \phi_i)_{i\in\cn} ,
\end{equation}
where $I_i$ and $\phi_i$ denote the current magnitude and current phase angle, respectively, of DER $i$. 
Then,~\cite{arbabi2018} augments the observable dictionary $\bm h$ with time delays and nonlinear functions: let $n_d\in\mathbb{N}^+$ denote the number of delays, the delay embedded vector $\bm\zeta_t$ stacks the most recent $n_d{+}1$ measurement and input snapshots
\begin{flalign}\label{equ:bm-zeta}
&\bm \zeta_t =
\left[
{\bm h(\bx_{t-n_d\Delta T^\mr}})^{\top}\; 
\cdots \; \bm h (\bx_{t})^\top \; \bu^{\top}_{t-n_d\Delta T^\mr} \; \cdots \; \bu^{\top}_{t}
\right]^\top, \hspace*{-5cm}\notag\\
& &&t\in\{t> n_d\Delta T^\mr\}\cap\ct^\mr_\mtrn.
\end{flalign}
We then lift $\bm \zeta$ using the nonlinear observable dictionary ${\bm g}(\bm \zeta)$, which depends meaningfully on $\bu_{t-n_d\Delta T^\mr},\ldots,\bu_{t}$, enabling EDMD to approximate the dynamics of the extended state space and discern the effect of previous inputs on state evolution.
In our implementation, we follow Example~1 of \cite{arbabi2018} to choose $n_d=5$ and $ {\bm g}(\bm \zeta) = [\bm \zeta^\top,~ \tfrac{1}{n_{\bm \zeta}}\|\bm \zeta\|_2, ~ 1]^\top$.

The conventional SINDy~\cite{nandakumar2023}, i) constructs the library $\btheta^{\mathrm{int}}$ using intuitively developed polynomial–trigonometric functions, ii) employs the original SINDy algorithm (Code~1 in~\cite{brunton2016a}) for identification, and iii) uses the same numerical method \eqref{equ:sindy-discrete} for prediction. Its library is defined as
\begin{equation}\label{equ:library-conventional}
\btheta^{\mathrm{int}}(\bx,\bu)\ceq \begin{bmatrix}
1; \bm{P}_2(\bx,\bu); \sin(\bx,\bu); \cos(\bx,\bu)\\
\end{bmatrix}.  
\end{equation}

Fig.~\ref{fig:4bus-prediction} presents the predicted frequencies of GFM~1 and GFL~2. 
The noisy PMU-measured frequencies ($f_1$ and $f_2$) are shown as solid blue lines, while the predictions obtained from $\operatorname{PC-SINDy}$, {EDMD-based Koopman approximation, and conventional SINDy are shown as yellow dash-dotted, green dotted, and pink dashed lines, respectively.} 
{All three methods} predict the frequencies accurately when the validation scenario is under the same configuration as the training one, i.e., for $t\in[10,10.5)$~s. 
However, starting from the first large transient at $t=10.5$~s, {both the approximated Koopman operator and }the conventional SINDy fail to capture the system dynamics, driving the predicted frequency outside the scale of Fig.~\ref{fig:4bus-prediction}. 
In contrast, PC-SINDy accurately tracks the true measurements even during large transients that are not present in the training data, demonstrating its strong generalization capability.

\subsection{Frequency Control Results by PC-SINDYc}\label{sec:sim-control}

\begin{figure}[htbp]
    \centering
    \includegraphics[width=\linewidth]{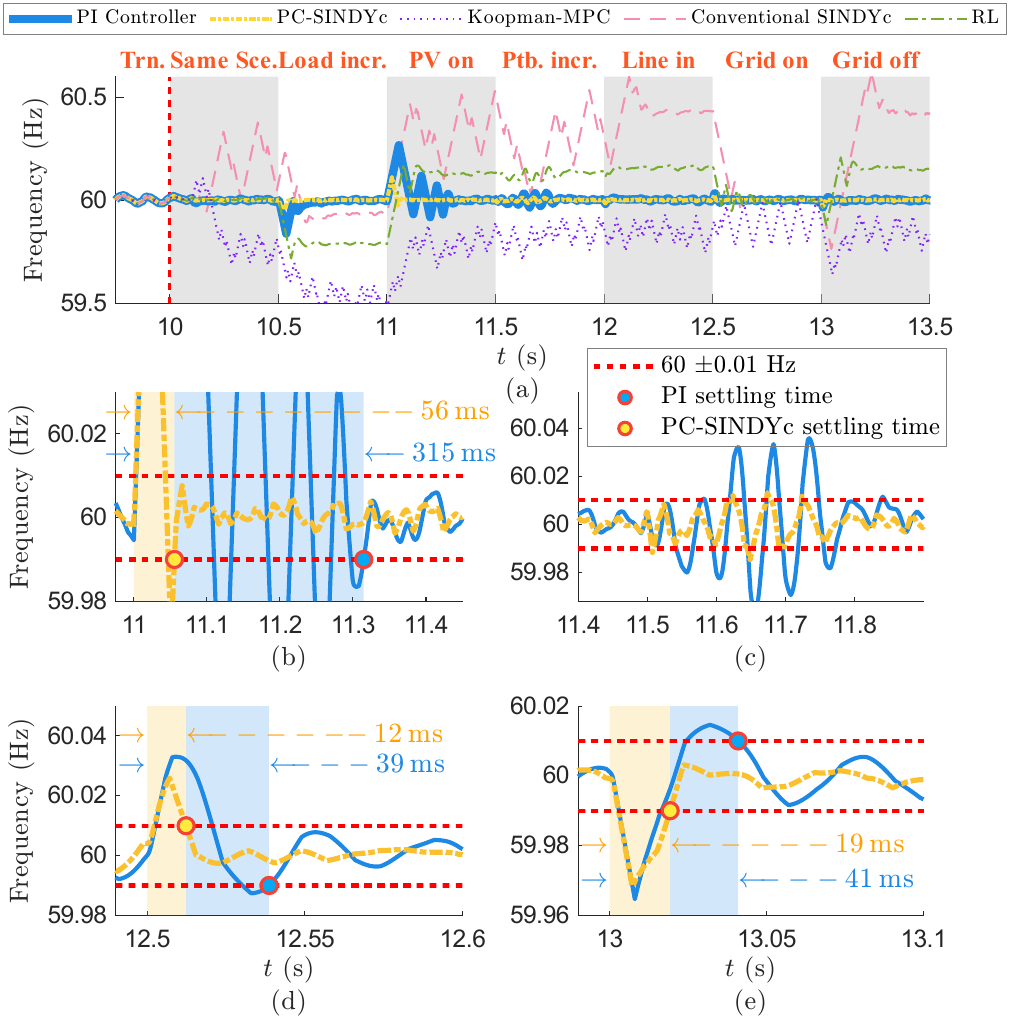}
    \caption{{GFM 1 frequency $f_1$ under various secondary controllers: (a) Overall transient response; zoomed-in views of transients during (b) PV~1 connecting to the MG, (c) load variation, (d) connecting to the distribution grid, and (e) grid disconnection.}}
    \label{fig:4bus-control}
\end{figure}
Fig.~\ref{fig:4bus-control} shows the control performance on the 4-bus MG under the same transient as Section~\ref{sec:sim-prediction}. 
We compare the proposed PC-SINDYc (dash-dotted yellow line in Fig.~\ref{fig:4bus-control}) with {four} benchmark methods: a PI controller~\cite{savaghebi2012} (solid blue line), {the MPC based on Koopman approximation~\cite{arbabi2018} (dotted purple line), }the conventional SINDYc~\cite{nandakumar2023} (dashed pink line), and a state-of-the-art RL controller~\cite{rodriguez-martinez2025} (dash-dotted green line). 
We tune the PI controller {and the weighting matrices for Koopman-based MPC and conventional SINDYc }with our best effort. 

Regarding the RL controller, we implement the TD3 algorithm. To ensure a fair comparison with PC-SINDYc, the TD3 agent uses the same observable space with state and action vectors defined as $s_t = [\btheta_t;\; \bdx_t]$ and $a_t = \bu_t=[\omega^s_{1,t};\; p^s_{1,t}]$, respectively.
Each training episode lasts $T_\mtrn = 10$~s with control step $\Delta T^\mr = 1/120$~s to align with the PMU reporting rate. To match the small-perturbation training conditions used for PC-SINDYc in Section~\ref{sec:method-activation}, we initialize $\omega^s_{1,0} \sim \operatorname{Unif}[0.999,\,1.001]$~p.u. and $p^s_{1,0} \sim \operatorname{Unif}[0.49,\,0.51]$~p.u., and update the actions through small Gaussian increments $\Delta \omega^s_{1,t} \ceq \omega^s_{1,t} - \omega^s_{1,t-\Delta T^\mr} \sim \mathcal{N}(0,\,0.001^2\Delta T^\mr/T_\mtrn)$~p.u. and $\Delta p^s_{1,t} \ceq p^s_{1,t} - p^s_{1,t-\Delta T^\mr} \sim \mathcal{N}(0,\,0.01^2\Delta T^\mr/T_\mtrn)$~p.u., so that the exploratory action trajectories remain small, continuous, and oscillatory throughout each episode. The reward function is defined as $R_t = -\big(\alpha_\omega (\Delta \omega_{1,t})^2 + \alpha_{\omega^s} (\Delta \omega^s_{1,t})^2 + \alpha_{p^s} (\Delta p^s_{1,t})^2\big)$, where $\Delta \omega_{1,t} \ceq \omega_{1,t} - \omega_0$, $\alpha_\omega = 5\times10^5$ p.u., and $\alpha_{\omega^s} = \alpha_{p^s} = 2$ p.u., so that the agent penalizes both frequency deviation and aggressive control variations. We train TD3 for $1000$ episodes in Simulink and evaluate it periodically during training. The detailed TD3 architecture, hyperparameter settings, and training procedure follow \cite{rodriguez-martinez2025}.

For validation, each controller is activated at \(t = 10\)~s and operates with a time step of \(\Delta T^{\mr}\). 
At each time step, the control signal experiences a delay of up to \(9.33\)~ms before the GFM converter acts, accounting for data acquisition~\cite{IEEE2011} and communication latency~\cite{hasan2021}.

Fig.~\ref{fig:4bus-control}(a) shows that both the proposed PC-SINDYc and the PI controller successfully restore and stabilize the frequency with minimal overshoot. 
In contrast, the {Koopman-based MPC and }conventional SINDYc fails to maintain stability due to {their} poor prediction performance demonstrated in the previous subsection. 
{Although both Koopman-based MPC and conventional SINDYc  have demonstrated strong capability in modeling nonlinear dynamics and control, their performance depends strongly on the training conditions. Since they are trained only on small-signal disturbance scenarios, as is PC-SINDYc, they generalize poorly when the system undergoes the large transients considered here. The resulting prediction errors lead to ineffective control.}
Similarly, the RL controller maintains the frequency well over \(t \in [10,10.5)\)~s, when the system configuration stays the same as the training one, indicating that the controller is well-tuned. However, once the unseen large-signal transient occurs at \(t = 10.5\)~s, the learned policy fails to handle the new scenario.

Figs.~\ref{fig:4bus-control}(b)–(e) provide zoomed-in views of key transients, including PV~1 connecting to the MG, load-perturbation intensity increasing by a factor of 10, connecting to the distribution grid, and subsequent disconnection. As seen in Fig.~\ref{fig:4bus-control}(b), PC-SINDYc restores the frequency more smoothly and rapidly than the PI controller, with smaller overshoot and a shorter settling time (PC-SINDYc settles in \(56\)~ms vs. \(315\)~ms for PI). 
In Fig.~\ref{fig:4bus-control}(c), when the stochastic load amplitude increases by 10 times, PC-SINDYc consistently maintains smaller frequency oscillations than the PI controller. Similarly, in Figs.~\ref{fig:4bus-control}(d) and (e), PC-SINDYc yields both smaller overshoot and shorter settling times compared with the PI controller.

\subsection{The Performance of PC-SINDYc in the 13-Bus MG}\label{sec:sim-13-bus} 

\begin{figure}[htbp]
    \centering
    \includegraphics[width=\linewidth]{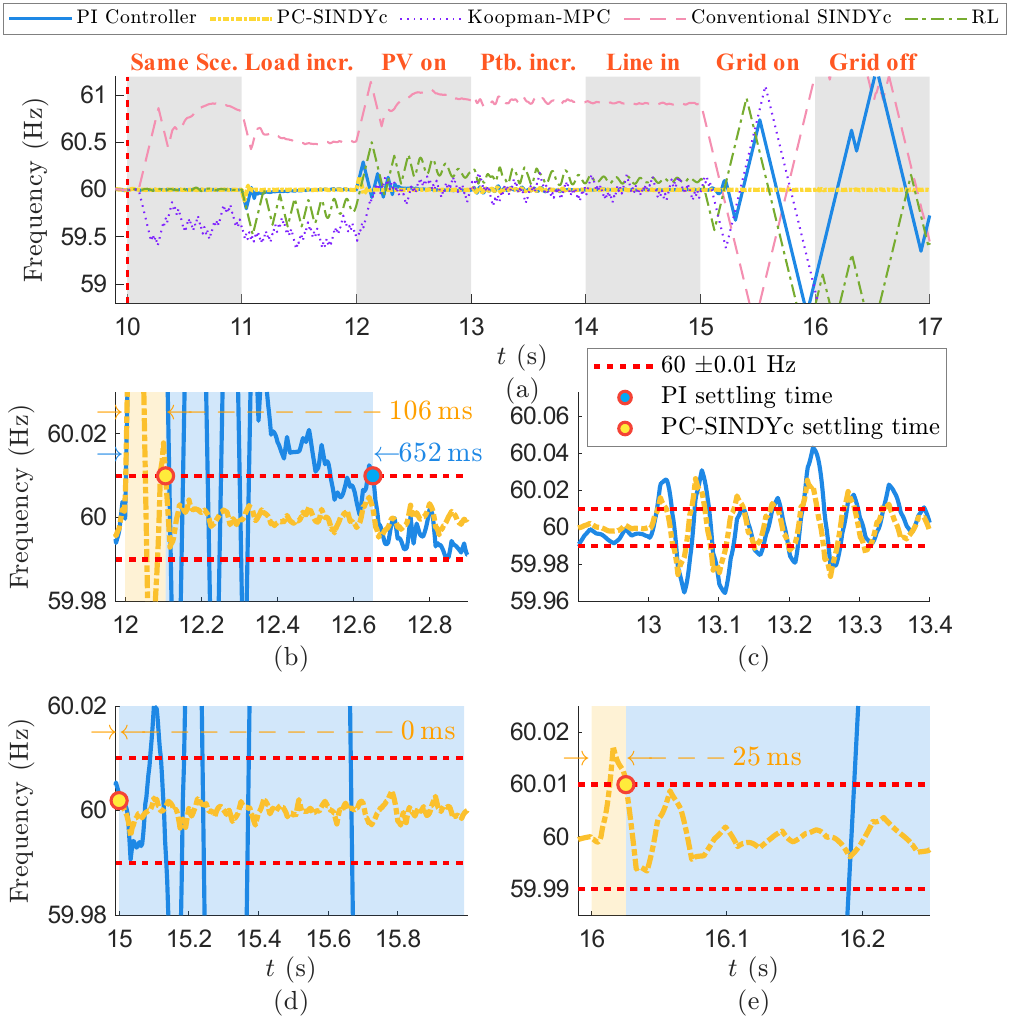}
    \caption{
    {GFM 1 frequency $f_1$ under various secondary controllers: (a) Overall transient response; zoomed-in views of transients during (b) PV~1 connecting to the MG, (c) load-perturbation intensity increasing, (d) connecting to the distribution grid, and (e) grid disconnection.}}
    \label{fig:13bus-control}
\end{figure}
We then compare our proposed PC-SINDYc with the same set of benchmark methods in a medium-voltage (4.16~kV) 13-bus MG with transformers and diesel generators, whose high inductances and inertia create a more complex electromagnetic environment. 
In this 13-bus MG, GFM~1 operates in VSG mode and GFLs~2–4 use single-phase PLLs, to demonstrate the generalization capability of PC-SINDYc.
{We apply the same transients as in Section~\ref{sec:sim-prediction} to the 13-bus MG, except that the duration of each transient is increased from \(0.5\)~s to \(1\)~s.}

Fig.~\ref{fig:13bus-control} illustrates control performance in the 13-bus MG. Similar to the 4-bus case, {Koopman-based MPC, }conventional SINDYc, and the RL controller fail to maintain frequency stability. 
The PI controller restores system frequency during load changes but exhibits larger overshoot, wider oscillation range, and longer settling time than PC-SINDYc (Figs.~\ref{fig:13bus-control}(b)–(c)).
When the MG connects to the distribution grid at \(t = 15\)~s (Fig.~\ref{fig:13bus-control}(d)), PC-SINDYc responds quickly to the new operating condition and maintains the frequency within \(60 \pm 0.01\)~Hz with negligible overshoot. In contrast, the PI controller resonates with the distribution grid and drives system frequency farther from nominal.
When the MG disconnects from the AC grid at \(t = 16\)~s (Fig.~\ref{fig:13bus-control}(e)), PC-SINDYc rapidly rebalances power flow within the MG, achieving smaller overshoot and settling time below \(25\)~ms, whereas the PI controller fails to restore frequency from preceding oscillations.

In summary, the proposed PC-SINDYc framework demonstrates strong generalization capability and robust control performance, achieving smaller overshoot and faster settling times even under unseen large disturbances.

\subsection{DC Dynamics}\label{sec:sim-dc-dynamics}

In practice, DC dynamics can be categorized into two types: small-amplitude but persistent DC-voltage oscillations caused by the DC-voltage control blocks, and large-amplitude, occasional, but sustained power intermittency caused by the stochastic nature of energy sources. 
Persistent DC-voltage oscillations will be reflected in the library, e.g., through $p_i$ for GFM $i\in\cn_\mgfm$ and the phase error $\delta_i$ for GFL $i\in\cn_\mgfl$, but they generally do not trigger the constraints \eqref{equ:constraint} because the constraint sets are intentionally chosen to be sufficiently large during normal operation.
In contrast, power intermittency may cause rapid changes in $p_i$, leading to fast frequency deviations and potential constraint activation. 

During system identification, we inject only small-amplitude oscillatory control signals \eqref{equ:excitation} into the MG over the 10-second excitation period (Section~\ref{sec:method-activation}). As a result, the training data rarely contain constraint-activation events. If power intermittency occurs during this period and activates the constraints \eqref{equ:constraint}, Algorithm~\ref{alg:ca-sindy} detects and rejects the resulting outliers, thereby ensuring accurate estimation of $\Xi$.

During control, power intermittency in GFLs does not affect frequency regulation, as only GFMs participate in control (Section~\ref{sec:method-mpc}). Power intermittency in GFMs is also rare, as they are typically supported by sufficient power reserves. In extreme cases where a GFM reaches its power limit, frequency regulation cannot be maintained by any control methods due to insufficient source power. 

To demonstrate these, we conduct simulations on the 4-bus MG in Fig.~\ref{fig:MG-circuit}(a). First, to emulate DC-voltage oscillations,
the constant DC sources are replaced with controlled sources whose voltages oscillate with amplitude 0.1~p.u. around 1~p.u. over the interval $t\in\ct_\mdc=[11,12]$~s:
\begin{align}
v^\mdc_{i,t}
&= 1+\psi(t;0.1,50,50\pi) &&i \in \cn,\ t \in \ct_\mdc.
\end{align}
Second, to model power intermittency, we reduce the source power of GFL~2 to $p_2^{\msrc}=0.1$~p.u. over $t\in[12,13]$~s, and the source power of GFM~1 to $p_1^{\msrc}=0.1$~p.u. over $t\in[13,14]$~s.
\begin{figure}[htbp]
    \centering
    \includegraphics[width=\linewidth]{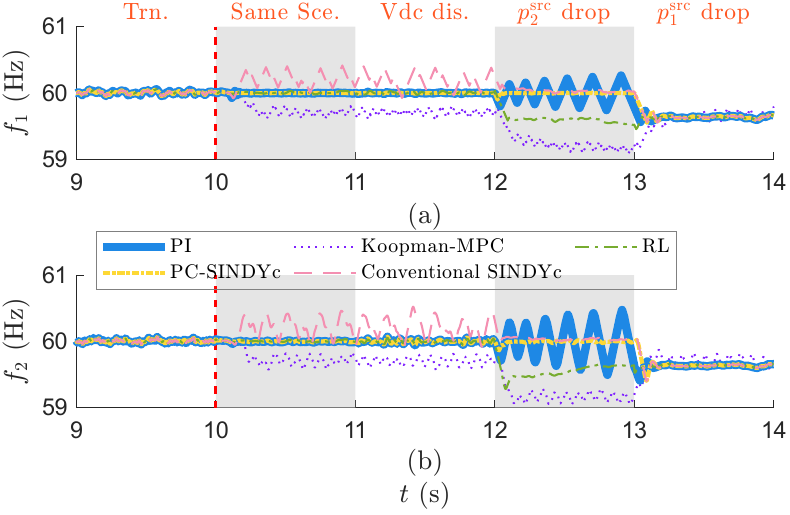}
    \caption{Frequency regulation under unseen DC dynamics: (a) GFM frequency $f_1$ and (b) GFL frequency $f_2$ under PC-SINDYc, PI control, {Koopman-based MPC, }conventional SINDYc, and RL control.}
    \label{fig:resp-3.2-control}
\end{figure}

Fig.~\ref{fig:resp-3.2-control} (a) and (b) compares PC-SINDYc, PI control, {Koopman-based MPC, }conventional SINDYc, and RL control under the same DC transients. Like PC-SINDYc, the RL controller is trained only on small-signal disturbances and does not encounter these DC dynamics during training. The controllers are activated at $t=10.1$~s. Over $t\in[10,11]$~s, {Koopman-based MPC and }conventional SINDYc already drives the system frequency to oscillate, even though the disturbance remains close to the training condition. During the DC-voltage oscillation interval $t\in[11,12]$~s, both PI and RL controllers regulate frequency satisfactorily. However, when the source power of GFL~2 drops over $t\in[12,13]$~s, the RL controller can no longer stabilize the frequency. The PI controller also begins to oscillate in this interval and drives the system toward instability. The triangle-wave pattern in the frequency traces is caused by activation of the frequency rate limiters. In contrast, PC-SINDYc enables GFM 1 to provide stronger support over $t\in[12,13]$~s and maintain stability, whereas PI control induces larger power oscillations.

When the source power of GFM~1 drops over $t\in[13,14]$~s, GFM~1 can no longer provide enough power to restore the frequency to the nominal value of 60~Hz. As a result, the system frequency declines, and none of the control methods can recover it during the power shortage. %

Overall, these results show that PC-SINDYc remain effective under unseen DC dynamics, including DC-voltage oscillations and power intermittency, until the available source power of the GFM itself becomes insufficient for frequency regulation.

\subsection{PC-SINDYc Performance in a MG with Multiple GFMs and GFLs, with and without Missing PMUs}\label{sec:sim-multiple-ders}

To demonstrate that the proposed PC-SINDYc coordinates effectively among multiple DERs in an MG with many GFMs and GFLs, we consider a modified WSCC 9-bus MG \cite{heins2023}, with 3 GFM converters installed at buses 1, 2, and 3 and 5 GFL converters installed at buses 4--8, as shown in Fig.~\ref{fig:resp-2.8-grid}.

\begin{figure}[htbp]
    \centering
    \includegraphics[width=\linewidth]{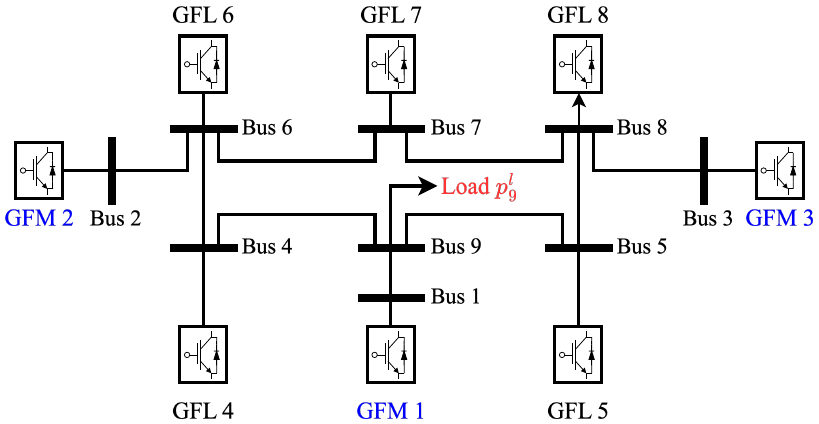}
    \caption{Modified WSCC 9-bus MG with 3 GFMs (buses 1--3) and 5 GFLs (buses 4--8).}
    \label{fig:resp-2.8-grid}
\end{figure}

The dynamic model is identified from the first 10~s of PMU measurements collected under small-signal excitations.
At $t = 10.5$~s, the load $p^l_{9}$ at bus~9 (red in Fig.~\ref{fig:resp-2.8-grid}) steps up by $1$~p.u., creating a large-signal disturbance.
\begin{figure}[htbp]
    \centering
    \includegraphics[width=\linewidth]{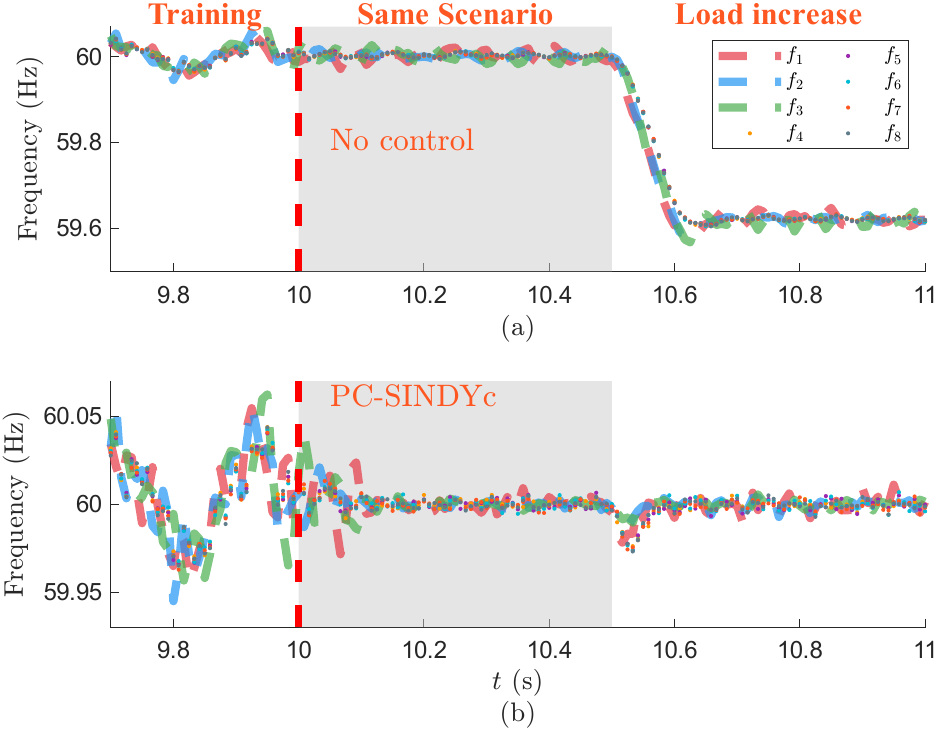}
    \caption{Frequency responses of the modified WSCC 9-bus MG under a $1$~p.u. load step at bus~9: (a) without control and (b) with PC-SINDYc activated at $t = 10.1$~s. The GFM frequencies $f_1$--$f_3$ are shown with dashed lines, and the GFL frequencies $f_4$--$f_8$ are shown with dotted lines.}
    \label{fig:resp-2.8-8DERs}
\end{figure}

Fig.~\ref{fig:resp-2.8-8DERs}(a) shows the frequency response without control, where the GFM frequencies $f_1$, $f_2$, and $f_3$ are plotted with dashed lines and the GFL frequencies $f_4$--$f_8$ with dotted lines. Without control, the system frequency drops to about 59.6~Hz.
Fig.~\ref{fig:resp-2.8-8DERs}(b) shows the frequency response when the MPC controller is activated at $t = 10.1$~s. With PC-SINDYc activated, only very small oscillations appear during the large-signal transient. The three GFM converters maintain stable frequency regulation with small oscillations, which demonstrates the scalability of the proposed framework.

\subsubsection{With Missing PMUs}\label{sec:sim-missing-pmu}
In real-world implementations, not every DER may be equipped with a PMU.
If DER $i$ lacks a PMU at its PCC, we cannot identify its dynamic model.
However, PC-SINDYc can still control the PMU-equipped DERs in the MG, provided at least one GFM is equipped with a PMU.

We reuse the modified WSCC 9-bus MG shown in Fig.~\ref{fig:resp-2.8-grid}, where PMUs are missing at GFM~3 and GFLs~4 and~7.
Following the same procedure as above, we apply the same large-signal load increase at $t=10.5$~s and activate PC-SINDYc on GFMs~1 and~2 at $t=10.1$~s for frequency regulation.

\begin{figure}[htbp]
    \centering
    \includegraphics[width=\linewidth]{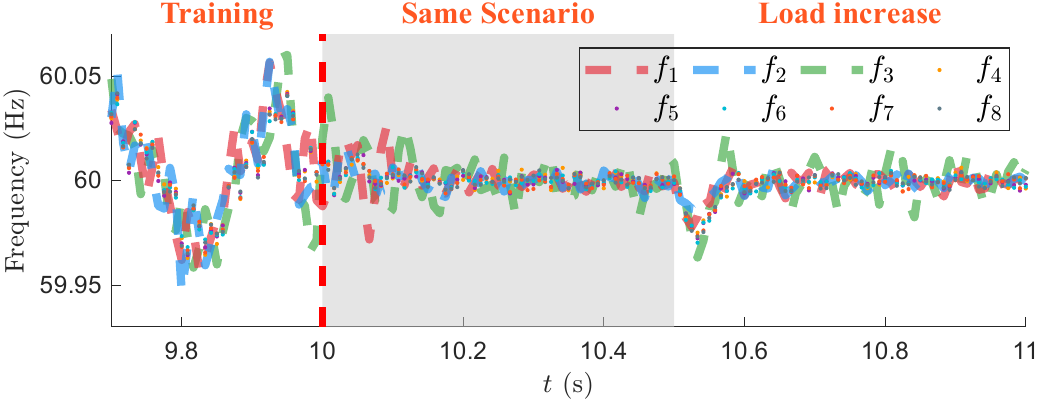}
    \caption{
    Closed-loop frequency response of the modified WSCC 9-bus MG under PC-SINDYc with missing PMUs}
    \label{fig:resp-1.7-8DERs}
\end{figure}

After the load increase at $t = 10.5$~s, the frequencies of GFMs~1 and~2, $f_1$ and $f_2$ (dashed red and blue), quickly restore the system frequency to 60~Hz.
The frequency of GFM~3, $f_3$ (dashed green), also returns to 60~Hz, but with slightly larger oscillations than in Fig.~\ref{fig:resp-2.8-8DERs}(b) because it does not participate in frequency control.
The GFL frequencies $f_4$--$f_8$ (dotted) also gradually return to the nominal value, despite the absence of PMUs at GFLs 4 and 7.  %

Overall, GFMs~1 and~2 maintain stable frequency regulation, confirming that PC-SINDYc remains effective under incomplete PMU coverage.

\section{Conclusion}\label{sec:conclusion}
This paper proposed PC-SINDYc, a novel framework for identifying the true frequency-dynamics model of DER-based MGs and performing real-time frequency control. The framework requires only practical PMU data and does not rely on prior knowledge of network parameters. We proved that, under mild conditions, PC-SINDYc ensures asymptotic stability of MGs. Case studies on 4-bus, 9-bus, and 13-bus MGs demonstrated the accuracy of PC-SINDYc in identifying the true frequency-dynamics models. Furthermore, PC-SINDYc achieved effective and stable frequency control under various disturbances beyond the offline identification/training conditions. 

{
The main limitations of the proposed PC-SINDYc are as follows. First, PMU deployment remains a key requirement: each DER must be equipped with a PMU to be both identifiable and controllable (although PC-SINDYc remains applicable to PMU-equipped DERs even when PMUs are missing for some DERs; see Section~\ref{sec:sim-missing-pmu}). Second, the method relies on access to high-frequency internal PMU samples (Section~\ref{sec:PMU}). Third, the stability guarantee in Theorem~\ref{theorem:1} is established for the identified model rather than the original MG dynamics (Remark~\ref{remark:theorem-limitation}).
Nevertheless, the case studies show that the proposed method remains effective under practical PMU noise, partial PMU availability (Section~\ref{sec:sim-missing-pmu}), 
and moderate model mismatch. Future work will extend the framework to voltage control and develop a more rigorous stability analysis in the presence of model mismatch, as discussed in Remark~\ref{remark:theorem-limitation}.
}
\appendices

\section{{PLL Equivalence}} \label{app:PLL-equiv}
In this appendix, we prove the equivalence between \eqref{equ:3p-pll-vq-vq} and \eqref{equ:1p-pll-vq-vq} 
under practical conditions. 
Recall equations \eqref{equ:3p-pll-vq-vq} and \eqref{equ:1p-pll-vq-vq}:
\begin{subequations}
\begin{align}
& \delta^{{\mtp}} = \mathbf{T}_{q}({\theta})\, \bm{v}^{\ml abc}, \tag{\ref{equ:3p-pll-vq-vq}} \\
& \delta^{{\msp}} (t) = -\frac{2}{T(t)}\int_{t-T(t)}^{t} v^{\ml a}(\tau) \sin(\theta(\tau)) \, \md\tau,  \tag{\ref{equ:1p-pll-vq-vq}}
\end{align}
\end{subequations}
where for brevity, we drop the subscript $(\cdot)_i$ and use superscripts $(\cdot)^{{\mtp}}$ and $(\cdot)^{{\msp}}$ to distinguish the three-phase PLL from the single-phase PLL. 

Equation~\eqref{equ:3p-pll-vq-vq} is equivalent to \eqref{equ:1p-pll-vq-vq} when the the three-phase voltage $\bv^{abc}$ is balanced and the PLL operates in the lock-in region. We formalize these conditions in Assumptions~\ref{assum:balanced} and \ref{assum:quasi-static}:
\begin{assumption}[Balanced voltage]\label{assum:balanced}
The three-phase voltage $\bv^{abc}$ is balanced:
\begin{equation}\label{equ:assum-balanced}
    \bm{v}^{abc} = V_m\big[\cos(\theta_0),\ \cos\!\left(\theta_0 - \tfrac{2\pi}{3}\right),\ \cos\!\left(\theta_0 + \tfrac{2\pi}{3}\right)\big]^\top,
\end{equation}
where $V_m > 0$ denotes the peak amplitude and $\theta_0(t)$ denotes true phase angle.
\end{assumption}
\begin{assumption}[Lock-in region]\label{assum:quasi-static}
The PLL operates in the lock-in region. In this regime, the phase error $\Delta\theta \ceq \theta_0-\theta$ remains within a small range with no oscillations for at least one slipping cycle (Section 8.2.3 of \cite{gardner2005}). Consequently, $\Delta\theta$ can be regarded as a constant over the slipping cycle $[t-T(\theta(t)),\,t]$: 
\begin{equation}\label{equ:assum-quasi-static}
    \Delta\theta(\tau) = \Delta\theta(t), \quad \forall\, \tau \in [t-T(\theta(t)),\,t].
\end{equation}
\end{assumption}

\begin{proposition}\label{prop:pll-equiv}
    Under Assumptions~\ref{assum:balanced}--\ref{assum:quasi-static}, 
    \eqref{equ:3p-pll-vq-vq} is equivalent to \eqref{equ:1p-pll-vq-vq}, 
    i.e., $\delta^{{\mtp}}(t) = \delta^{{\msp}}(t)$.
\end{proposition}

\begin{proof}
Regarding the three-phase PLL, we have
\begin{flalign}
    \delta^{\mtp}
        &= \mathbf{T}_q({\theta})\,\bm{v}^{abc} \hspace*{-5cm} &&\nonumber \\
        &\overset{\ref{assum:balanced}}{=}
            -\frac{2V_m}{3}\!\bigg(
                \sin(\theta)\cos(\theta_0) \nonumber + \sin\!\left(\theta-\tfrac{2\pi}{3}\right)\cos\!\left(\theta_0-\tfrac{2\pi}{3}\right) \hspace*{-5cm} && \nonumber \\
        & && + \sin\!\left(\theta+\tfrac{2\pi}{3}\right)\cos\!\left(\theta_0+\tfrac{2\pi}{3}\right)
            \bigg) \nonumber \\
        &\overset{\eqref{requ:product2sum}}{=}
            -\frac{V_m}{3}\bigg(
                \sin(\theta_0+\theta)+\sin(\theta_0-\theta) \hspace*{-5cm} \nonumber \\
        &\hspace{6em}
                +\sin\!\left(\theta_0+\theta-\tfrac{4\pi}{3}\right)+\sin(\theta_0-\theta) \hspace*{-5cm} \nonumber \\
        &\hspace{9em}
                +\sin\!\left(\theta_0+\theta+\tfrac{4\pi}{3}\right)+\sin(\theta_0-\theta)
            \bigg) \hspace*{-5cm} \nonumber \\
        &= -\frac{V_m}{3}\bigg(
                \underbrace{
                    \sin(\theta_0\!+\!\theta)
                    \!+\!\sin\!\left(\theta_0\!+\!\theta-\tfrac{4\pi}{3}\right)
                    \!+\!\sin\!\left(\theta_0\!+\!\theta\!+\!\tfrac{4\pi}{3}\right)
                }_{=\,0} \hspace*{-5cm} \nonumber \\
        &\hspace{17em}
                + 3\sin(\theta_0-\theta)\bigg) \hspace*{-5cm} \nonumber \\
        &= -V_m\sin(\Delta\theta), \hspace*{-5cm} \label{eq:delta-tp-final}
\end{flalign}
where \eqref{requ:product2sum} is the product-to-sum identity
\begin{equation}\label{requ:product2sum}
    \cos\!A\sin\!B = \tfrac{1}{2}[\sin(A+B)+\sin(B-A)].
\end{equation}

Regarding the single-phase PLL, we have
    \begin{flalign}
        \delta^{\msp}(t) &\overset{\ref{assum:balanced}}{=}
            -\frac{2V_m}{T(t)}\int_{t-T(t)}^{t} \cos(\theta_0(\tau)) \sin(\theta(\tau)) \,\md \tau \hspace*{-5cm}&&\notag\\
        &= -\frac{V_m}{T(t)}\int_{t-T(t)}^{t}
            \Big[\sin\!\big(\theta_0(\tau)+\theta(\tau)\big) \hspace*{-5cm}&&\notag\\
        & &&+\sin\!\big(\theta_0(\tau)-\theta(\tau)\big)\Big]\,\md \tau \notag\\
        &= -\frac{V_m}{T(t)}\int_{t-T(t)}^{t} \sin\!\big(2\theta(\tau)+\Delta\theta(\tau)\big) \,\md \tau \hspace*{-5cm}&& \notag\\
        & &&
            -\frac{V_m}{T(t)}\int_{t-T(t)}^{t} \sin\!\big(\Delta\theta(\tau)\big) \,\md \tau \notag\\
        &\overset{{\ref{assum:quasi-static}}}{=}
            -\frac{V_m}{T(t)}\underbrace{\int_{t-T(t)}^{t} \sin\!\big(2\theta(\tau)+\Delta\theta({t})\big) \,\md \tau}_{=0} \hspace*{-5cm}&&\notag\\
        & && 
            -\frac{V_m}{T(t)}\underbrace{\int_{t-T(t)}^{t} \sin\!\big(\Delta\theta({t})\big) \,\md \tau}_{=T(t)\sin\!\big(\Delta\theta(t)\big)} \notag\\
        &= -V_m\sin(\Delta\theta(t)). \hspace*{-5cm}&&\label{eq:delta-sp-final}
    \end{flalign}

As such, we have 
\begin{equation}
    \delta^{\mtp}(t) \overset{\eqref{eq:delta-tp-final}}{=} -V_m\sin(\Delta\theta(t)) \overset{\eqref{eq:delta-sp-final}}{=}\delta^{\msp}(t),
    \label{eq:1p-eq-3p}
\end{equation}
which proves Proposition~\ref{prop:pll-equiv}.
\end{proof}

\bibliography{utils/SINDy}
\bibliographystyle{IEEEtran}

\end{document}

%% file: utils/preamble.tex
\usepackage{amsmath,amsfonts,amssymb,mathtools}
\usepackage{bm,upgreek}
\usepackage{amsthm}
\usepackage{physics}
\allowdisplaybreaks

\usepackage{graphicx,scalerel}
\usepackage{array,booktabs,multirow,colortbl}
\usepackage[caption=false,font=normalsize,labelfont=sf,textfont=sf]{subfig}
\usepackage[inkscapepath=./.out/svg-inkscape/]{svg}

\usepackage{algorithm}
\usepackage[noend]{algpseudocode}

\usepackage{enumitem}
\usepackage{xcolor}
\usepackage{tcolorbox}

\usepackage{textcomp,stfloats,url}
\usepackage{verbatim}
\usepackage{cite}

\usepackage[normalem]{ulem}
\usepackage[hidelinks]{hyperref}
\usepackage{orcidlink}

\usepackage[compact]{titlesec}
\titlespacing{\section}{0pt}{*0.4}{*0.4}
\titlespacing{\subsection}{0pt}{*0.3}{*0.3}
\titleformat{\subsubsection}[runin]
  {\normalfont\normalsize\itshape}
  {\arabic{subsubsection})}
  {0.5em}
  {}
\titlespacing*{\subsubsection}{2pt}{*0.3}{0.5em}

\setlist[enumerate]{itemsep=0.3ex, topsep=0.3ex}

\makeatletter
\def\thm@space@setup{%
  \thm@preskip=3pt
  \thm@postskip=3pt
}
\makeatother

%% file: utils/macros.tex
\newcommand{\Output}{\item[\textbf{Output:}]}
\newcommand{\Inputs}{\item[\textbf{Inputs:}]}

\makeatletter
\newcommand{\StateH}[1]{%
  \State\leavevmode\begingroup
  \setlength{\fboxsep}{0pt}%
  \colorbox{gray!20}{%
    \parbox{\dimexpr\linewidth-\ALG@thistlm\relax}{\strut#1}%
  }%
  \endgroup
}
\makeatother

\newcommand{\mint}{\mathrm{int}}

\newcommand{\hXi}{\hat{\Xi}}

\newcommand{\mload}{\mathrm{load}}
\newcommand{\msrc}{\mathrm{src}}
\newcommand{\mdc}{\mathrm{dc}}

\newcommand{\cceq}{\coloneqq}
\newcommand{\ceq}{:=}
\newcommand{\md}{\mathrm{d}}

\newcommand{\bx}{\bm{x}}
\newcommand{\bdx}{\dot{\bm{x}}}
\newcommand{\bu}{\bm{u}}
\newcommand{\by}{\bm{y}}

\newcommand{\be}{\bm{e}}
\newcommand{\br}{\bm{r}}

\newcommand{\bkappa}{\bm{\kappa}}

\newcommand{\btheta}{\bm{\vartheta}}

\newcommand{\bv}{\bm{v}}

\newcommand{\bF}{\bm{F}}

\newcommand{\bX}{\bm{X}}
\newcommand{\bdX}{\dot{\bm{X}}}
\newcommand{\bU}{\bm{U}}

\newcommand{\bTheta}{\bm{\Theta}}

\newcommand{\bQ}{\bm{Q}}
\newcommand{\bR}{\bm{R}}
\newcommand{\bP}{\bm{P}}

\newcommand{\tbx}{\tilde{\bx}}
\newcommand{\tx}{\tilde{x}}
\newcommand{\tbu}{\tilde{\bu}}
\newcommand{\tbtheta}{\tilde{\btheta}}
\newcommand{\tXi}{\tilde{\Xi}}
\newcommand{\tcx}{\tilde{\cx}}
\newcommand{\tcu}{\tilde{\cu}}
\newcommand{\tbF}{\tilde{\bF}}

\newcommand{\bbn}{\mathbb{N}}
\newcommand{\bfi}{\mathbf{I}}

\newcommand{\ct}{\mathcal{T}}
\newcommand{\cn}{\mathcal{N}}
\newcommand{\cx}{\mathcal{X}}
\newcommand{\cu}{\mathcal{U}}

\newcommand{\mr}{\mathrm{r}}
\newcommand{\ms}{\mathrm{s}}

\newcommand{\ml}{\mathrm{}}

\newcommand{\mtrn}{\mathrm{trn}}
\newcommand{\mgfm}{\mathrm{GFM}}
\newcommand{\mgfl}{\mathrm{GFL}}
\newcommand{\mvsg}{\mathrm{vsg}}
\newcommand{\mdroop}{\mathrm{droop}}
\newcommand{\mtp}{\mathrm{3p}}
\newcommand{\msp}{\mathrm{1p}}
\newcommand{\mref}{\mathrm{ref}}

\DeclareMathOperator{\diag}{diag}

\newtheoremstyle{remarkstyle}%
  {0pt}{0pt}{}{\parindent}%
  {\bfseries\itshape}{:}{ }%
  {\bfseries\itshape\thmname{#1}\thmnumber{ #2}\thmnote{\textemdash #3}}

\theoremstyle{remarkstyle}
\newtheorem{remark}{Remark}

\theoremstyle{plain}
\newtheorem{theorem}{Theorem}
\newtheorem{lemma}{Lemma}
\newtheorem{proposition}{Proposition}
\newtheoremstyle{assumstyle}%
  {0pt}{0pt}{}{\parindent}%
  {\bfseries}{.}{ }%
  {\thmnumber{#2}\thmnote{ (#3)}}
\theoremstyle{assumstyle}
\newtheorem{assumption}{A}

\newcounter{notes}

\definecolor{modifycolor}{HTML}{0070C0}